\documentclass[aps,prx,twocolumn,footinbib,superscriptaddress,floatfix]{revtex4-2}
\usepackage{graphicx}
\usepackage{indentfirst}
\usepackage{physics}
\usepackage{braket}
\usepackage{float}
\usepackage{mathtools}
\usepackage{epstopdf} 
\usepackage{footnote}
\usepackage{esint}
\usepackage{comment}
\usepackage{color}
\usepackage[T1]{fontenc}
\usepackage[caption=false,position=top]{subfig}
\usepackage{amsfonts}
\usepackage{footmisc}
\usepackage{scrextend}
\usepackage{multirow}
\usepackage[hyperfootnotes=false]{hyperref}
\usepackage[acronym]{glossaries}
\usepackage[english]{babel}
\usepackage{url}
\usepackage{bm}
\usepackage{algpseudocode}
\usepackage{svg}
\usepackage{etoolbox}
\usepackage{enumerate}
\usepackage{soul}
\definecolor{darkblue}{rgb}{0,0,0.5}
\hypersetup{
    colorlinks=true,
    linkcolor=black,
    filecolor=blue,
    citecolor=darkblue,  
    urlcolor=black,
}

\apptocmd{\sloppy}{\hbadness 9999\relax}{}{}

\newtheorem{theorem}{Theorem}

\newtheorem{lemma}[theorem]{Lemma}

\newenvironment{proof}[1][Proof]{\noindent\textbf{#1.} }{\ \rule{0.5em}{0.5em}}

\newcommand{\calB}{{\cal B}}

\newcommand{\calD}{{\cal D}}
\newcommand{\calE}{{\cal E}}

\newcommand{\calL}{{\cal L}}
\newcommand{\calM}{{\cal M}} 
\newcommand{\calN}{{\cal N}} 
\newcommand{\calP}{{\cal P}} 
\newcommand{\calS}{{\cal S}}

\newcommand{\calH}{{\cal H}}

\newcommand{\1}{^{(1)}}

\newcommand{\bI}{\boldsymbol I}

\DeclareMathOperator{\E}{\mathbb{E}}

\def\be{\begin{equation}}
\def\ee{\end{equation}}
\def\ba{\begin{eqnarray}}
\def\ea{\end{eqnarray}}

\begin{document}
\title{Generative Quantum Machine Learning via Denoising Diffusion Probabilistic Models}

\author{Bingzhi Zhang}
\affiliation{ Department of Physics and Astronomy, University of Southern California, Los
Angeles, California 90089, USA
}
\affiliation{
Ming Hsieh Department of Electrical and Computer Engineering, University of Southern California, Los
Angeles, California 90089, USA
}

\author{Peng Xu}
\affiliation{Department of Statistics, University of Illinois at Urbana-Champaign, Champaign, Illinois 61820, USA}

\author{Xiaohui Chen}
\email{xiaohuic@usc.edu}
\affiliation{Department of Mathematics, University of
Southern California, Los
Angeles, California 90089, USA}

\author{Quntao Zhuang}
\email{qzhuang@usc.edu}
\affiliation{
Ming Hsieh Department of Electrical and Computer Engineering, University of Southern California, Los
Angeles, California 90089, USA
}
\affiliation{ Department of Physics and Astronomy, University of Southern California, Los
Angeles, California 90089, USA
}

\begin{abstract}
Deep generative models are key-enabling technology to computer vision, text generation, and large language models. Denoising diffusion probabilistic models (DDPMs) have recently gained much attention due to their ability to generate diverse and high-quality samples in many computer vision tasks, as well as to incorporate flexible model architectures and a relatively simple training scheme.  
Quantum generative models, empowered by entanglement and superposition, have brought new insight to learning classical and quantum data. Inspired by the classical counterpart, we propose the \emph{quantum denoising diffusion probabilistic model} (QuDDPM) to enable efficiently trainable generative learning of quantum data. 
QuDDPM adopts sufficient layers of circuits to guarantee expressivity, while it introduces multiple intermediate training tasks as interpolation between the target distribution and noise to avoid barren plateau and guarantee efficient training.
We provide bounds on the learning error and demonstrate QuDDPM's capability in learning correlated quantum noise model, quantum many-body phases, and topological structure of quantum data. The results provide a paradigm for versatile and efficient quantum generative learning.
\end{abstract}
\maketitle

Variational parametrized quantum circuits (PQCs)~\cite{cerezo2021variational,killoran2019continuous,abbas2021power,cong2019quantum} provide a near-term platform for quantum machine learning~\cite{biamonte2017quantum,rebentrost2014quantum,lloyd2014quantum}. In terms of generative models~\cite{gao2022enhancing,khoshaman2018quantum,amin2018quantum,gao2018quantum}, quantum generative adversarial networks (QuGANs) have been recently proposed~\cite{lloyd2018quantum,dallaire2018quantum,hu2019quantum,huang2021experimental,zhu2022generative}, in analogy to classical generative adversarial networks (GANs)~\cite{NIPS2014_5ca3e9b1}.
Despite the success, classical GAN models are known for training issues such as mode collapse. In classical deep learning,
denoising diffusion probabilistic models (DDPMs) and their close relatives~\cite{sohl2015deep,ho2020denoising,song2021scorebased,schneuing2022structurebased,ChenGeorgiouPavon2016} have recently gained much attention due to relatively simple training schemes and their ability to generate diverse and high-quality samples in many computer vision tasks~\cite{dhariwal2021diffusion,muller2022diffusion,NEURIPS2021_7d6044e9,NEURIPS2019_3001ef25} over the best GANs, and to incorporate flexible model architectures. 

\begin{figure}
    \centering
    \includegraphics[width=0.45\textwidth]{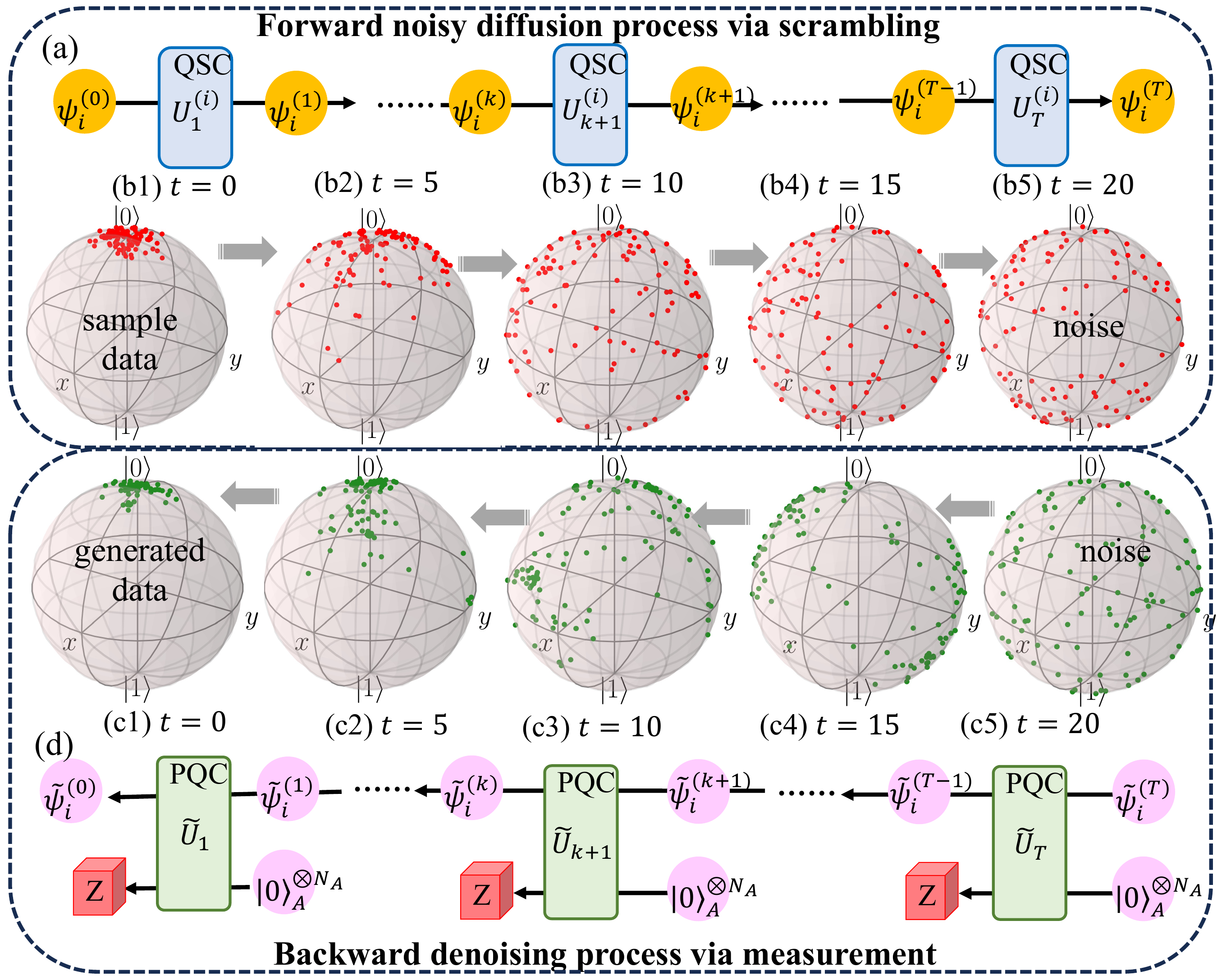}
    \caption{Schematic of QuDDPM. The forward noisy process is implemented by a quantum scrambling circuit (QSC) in (a), while in the backward denoising process is achieved via measurement enabled by ancilla and PQC in (d). 
    Subplots (b1)-(b5) and (c1)-(c5) present the Bloch sphere dynamics in generation of states clustering around $\ket{0}$, where convergence can be seen despite sample fluctuations, as shown in the Appendix~\ref{app:cluster}.
    }
    \label{fig:schematic}
\end{figure}

In this work, we propose the \emph{quantum denoising diffusion probabilistic model} (QuDDPM) as an efficiently trainable scheme to generative quantum learning, through a coordination between a forward noisy diffusion process via quantum scrambling~\cite{nahum2017quantum,nahum2018operator} and a backward denoising process via quantum measurement.
We provide bounds on the learning error and then demonstrate QuDDPM's capability in examples relevant to characterizing quantum device noises, learning quantum many-body phases and capturing topological structure of quantum data.
For an $n$-qubit problem, QuDDPM adopts linear-in-$n$ layers of circuits to guarantee expressivity, while it introduces $T\sim n/\log(n)$ intermediate training tasks to guarantee efficient training.

\section{General formulation of QuDDPM}
We consider the task of the generating new elements from an unknown distribution $\calE_0$ of quantum states, provided only a number of samples $\calS_0=\{\ket{\psi_k}\}\sim \calE_0$ from the distribution.  
The task under consideration---generating individual states from the distribution (e.g., a single Haar random state or $K$-design state)---is {\it not} equivalent to generating the average state of a distribution (e.g., a fully mixed state for Haar ensemble) considered in previous works of QuGAN~\cite{lloyd2018quantum}. To complete the task, QuDDPM learns a map from a noisy unstructured distribution of states to the structured target distribution $\calE_0$. It does so via a divide-and-conquer strategy of creating smooth interpolations between the target distribution and full noise, so that the training is divided to subtasks on a low-depth circuit to avoid barren plateau~\cite{mcclean2018barren, cerezo2021cost,wang2021noise,carlos2021entanglement}.

As shown in Fig.~\ref{fig:schematic}, QuDDPM includes two quantum circuits, one to enable the forward diffusion of sample data toward noise via scrambling and one to realize the backward denoising from noise toward generated data via measurement. For each data $\ket{\psi_i^{(0)}}$, the forward scrambling circuit [Fig.~\ref{fig:schematic} (a)] samples a series of $T$ random unitary gates $U_1^{(i)},\dots, U^{(i)}_T$ independently, such that the ensemble $\calS_k=\{\ket{\psi_{i}^{(k)}}=\prod_{\ell=1}^k U_\ell^{(i)} \ket{\psi_i^{(0)}}\}_i$ evolves from the sample data toward a random ensemble of pure states from $k=0$ to $k=T$. A Bloch sphere visualization of such a forward scrambling dynamics is depicted in Figs.~\ref{fig:schematic}(b1)-(b5) for a toy problem of learning single-qubit states $\calS_0$ clustered around a single pure state, e.g. $\ket{0}$, (b1), where the noise $\calS_T$ is uniform on the Bloch sphere (b5).

With the interpolation from data $\calS_0$ and noise $\calS_T$ in hand, the backward process can start from randomly sampled noise $\tilde{S}_T$ [Fig.~\ref{fig:schematic}(c5)] and reduce the noise gradually via measurement step by step, toward the final generated data $\tilde{\calS}_0$ [Fig.~\ref{fig:schematic}(c1)] that mimic the sample data [Fig.~\ref{fig:schematic}(b1)]. Measurements are necessary, as the denoising map is contractive and maintains the purity of all generated data in $\tilde{\calS}_0$.
As shown in Fig.~\ref{fig:schematic}(d), each denoising step adopts a unitary $\tilde{U}_k$ on the system plus $n_A$ number of ancilla qubits in $\ket{0}$ and performs a projective measurement in computational bases on the ancilla after the unitary $\tilde{U}_k$. Starting from the state $\ket{\tilde{\psi}_i^{(T)}}$, which is randomly sampled from noise ensemble, each unitary plus measurement step evolves the random state toward the generated data $\ket{\tilde{\psi}_i^{(0)}}$. Note that here all unitaries $\tilde{U}_k$ are fixed after training. In practice, the generation of noisy $\ket{\tilde{\psi}_i^{(T)}}$ can be directly completed by running the $T$ layers of the forward scrambling circuit on an arbitrary initial state. 
Via training, the denoising process learns information about the target from the ensembles in the forward scrambling, stores information in the circuit parameters, and then encodes onto the generated data.

\begin{figure}
    \centering
    \includegraphics[width=0.475\textwidth]{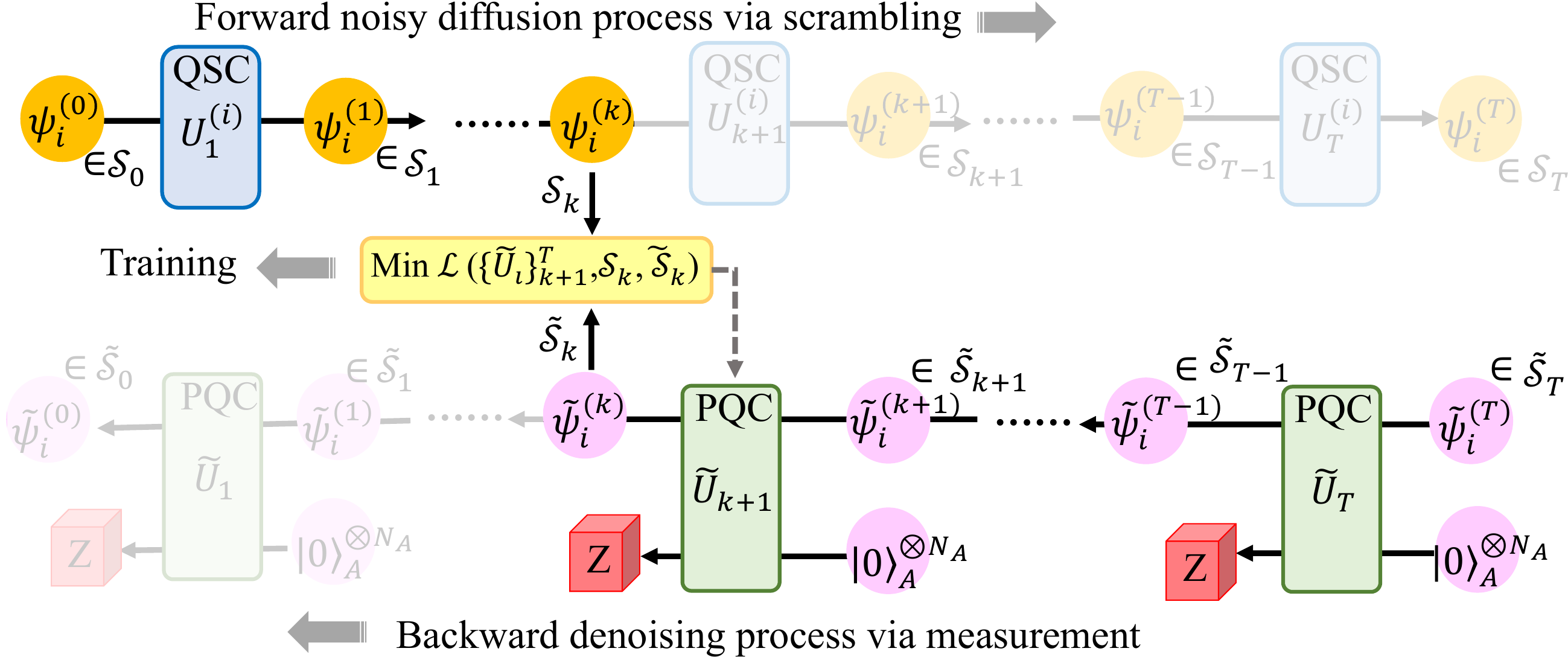}
    \caption{The training of QuDDPM at each step $t=k$. Pairwise distance between states in generated ensemble $\tilde{\psi}_i^{(k)}\in \tilde{\calS}_{k}$ and true diffusion ensemble $\psi_j^{(k)}\in\calS_k$ is measured and utilized in the evaluation of the loss function $\calL$. }
    \label{fig:schematic_train}
\end{figure}

\section{Training strategy}
In classical DDPM, the Gaussian nature of the diffusion allows efficient training via maximizing an evidence lower bound for the log-likelihood function, which can be evaluated {\it analytically}~\cite{sohl2015deep,ho2020denoising} (see Appendix~\ref{app:review}).
However, in QuDDPM, we do not expect such analytical simplification to exist at all---classical simulation of quantum device is inherently inefficient. Instead, the training of the QuDDPM relies on the capability of quantum measurements to extract information about the ensemble of quantum states for the efficient evaluation of a loss function.

The training of a $T$-step QuDDPM consists of $T$ training cycles, starting from the first denoising step $\tilde{U}_T$ toward the last $\tilde{U}_1$. As shown in Fig.~\ref{fig:schematic_train}, at the training cycle $(T+1-k)$, the forward noisy diffusion process is implemented from $U_1^{(i)}$ to $U_k^{(i)}$ to generate the noisy ensemble $\calS_k=\{\ket{\psi_i^{(k)}}\}_i$, While the backward denoising process performs the denoising steps $\tilde{U}_T$ to $\tilde{U}_{k+1}$ to generate the denoising ensemble $\tilde{\calS}_k=\{\ket{\tilde{\psi}_i^{(k)}}\}_i$. Within the training cycle, the parameters of the denoising PQC $\tilde{U}_{k+1}$ are updated such that the generated denoising ensemble $\tilde{\calS}_k$ converges to the noisy ensemble $\calS_k$. Therefore, QuDDPM divides the original training problem into $T$ smaller and easier ones. Indeed, with a local loss function, for $n$ qubits we can divide the $\Omega(n)$ layers (required by expressivity) of gates into $T\in \Omega(n/\log n)$ diffusion steps, such that each $\tilde{U}_{k+1}$ has order $\log(n)$ layers of gates to avoid barren plateau~\cite{cerezo2021cost}.

\section{Loss function}
To enable training, a loss function quantifies the distance between the two ensembles of quantum states. In this work, we focus on the maximum mean discrepancy (MMD)~\cite{JMLR:v13:gretton12a} (see Appendix~\ref{app:RKHS_MMD}) and the Wasserstein distance~\cite{villani2003topics,oreshkov2009distinguishability} (see Appendix~\ref{app:wasserstein_dist}) based on the state overlaps $\lvert\braket{\psi_i^{(k)}|\tilde{\psi}_i^{(k)}}\rvert^2$ estimated via a swap test (see Appendix~\ref{app:swap}).

Given two independent distributions of pure states $\calE_1$ and $\calE_2$ on the state vector space $V$, the \emph{statewise fidelity} between $\ket{\psi}$ and $\ket{\phi}$ is defined as $F(\ket{\phi}, \ket{\psi}) = \lvert\braket{\phi|\psi}\rvert^2$, and we can further define the \emph{mean fidelity} as
\begin{equation}
    \label{eqn:fidelity}
    \overline{F}(\calE_1, \calE_2)=\E_{\ket{\phi}\sim\calE_1,\ket{\psi}\sim\calE_2}\left[\lvert\braket{\phi|\psi}\rvert^2\right],
\end{equation}
where the random states $\ket{\phi} \sim \calE_1$ and $\ket{\psi} \sim \calE_2$ are drawn independently. Since the fidelity $F$ is a symmetric and positive definite quadratic kernel, according to the theory of reproducing kernel Hilbert space~\cite{scholkopf2018}, the MMD distance can be written as (see Appendix~\ref{app:RKHS_MMD})
\begin{align}
    \calD_\text{MMD}(\calE_1, \calE_2) = \overline{F}(\calE_1, \calE_1) + \overline{F}(\calE_2, \calE_2)-2\overline{F}(\calE_2, \calE_1),
    \label{eq:nat_distance}
\end{align}
which allows the estimation of MMD through sampled state ensembles $\calS_1$ and $\calS_2$. The expressivity of general MMD as a statistical distance measure depends on the kernel. On one hand, identifiability requires that the distance be zero if and only if $\calE_1 = \calE_2$. On the other hand, one also needs to ensure the quality of statistical estimation of the distance with a finite sample size of state ensembles. Hence, whether fidelity~\eqref{eqn:fidelity} is a proper kernel choice is problem-dependent. In Appendix~\ref{app:distance_degenerate}, we show an example where $\calD_\text{MMD}$ in Eq.~\eqref{eq:nat_distance} fails to distinguish two simple distributions. 
To resolve such an issue, we may alternatively consider the Wasserstein distance, a geometrically meaningful distance for comparing complex data distributions based on the theory of optimal transport~\cite{villani2003topics,oreshkov2009distinguishability} (see the Appendix~\ref{app:wasserstein_dist} for details). 

\begin{figure}
    \centering
    \includegraphics[width=0.45\textwidth]{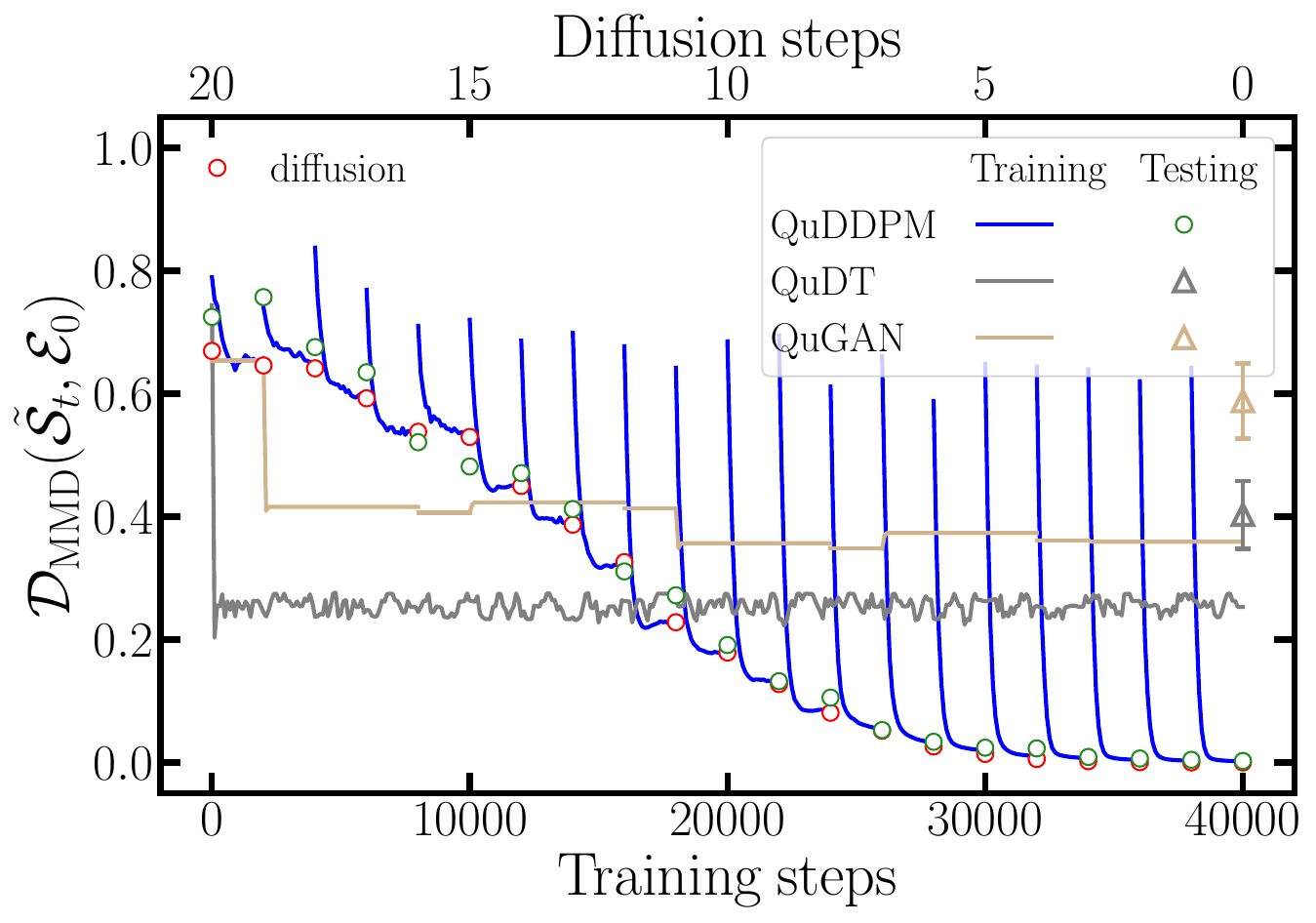}
    \caption{The decay of MMD distance $\calD$ between generated ensemble $\tilde{\calS}_t$ using different models and target ensemble of states $\calE_0$ clustered around $\ket{0,0}$ versus training steps. The converged value is $\calD \simeq 0.002$ for QuDDPM, showing an advantage of $2$ orders of magnitude over QuDT and QuGAN.}
    \label{fig:MMD_cost_cluster_two_qubit}
\end{figure}

As shown in Fig.~\ref{fig:schematic_train}, in the training cycle $t=k$, loss is a function of the unitary $\{\tilde{U}_\ell\}_{k+1}^T$ and depends on the noise distribution $\tilde{\calE}_T$ and the scrambled data distribution $\calE_k$,
\be 
\calL(\{\tilde{U}_\ell\}_{k+1}^T,\calE_k, \tilde{\calE}_T)=\calD\left(\calE_k, \tilde{\calE}_k[\{\tilde{U}_\ell\}_{k+1}^T,\tilde{\calE}_T]\right),
\ee
where $\calD$ can be the MMD distance or the Wasserstein distance. The distribution $\tilde{\calE}_k$ is a function of all the reverse denoising steps from $T$ to $k+1$ and the noise distribution $\tilde{\calE}_T$. In practice, we use finite samples to approximate the loss function as $\calL(\{\tilde{U}_\ell\}_{k+1}^T,\calS_k, \tilde{\calS}_T)$.

The toy example in Fig.~\ref{fig:schematic} adopted the MMD distance in the loss function and details of the training can be found in Appendix~\ref{app:cluster} (see Ref.~\cite{Github} for codes and data, and Table~\ref{tab:hyperparams} for details of parameters). Here we present the training history of a more challenging $2$-qubit example of preparing states clustered around $\ket{0,0}$, to allow a meaningful comparison with other algorithms. 
In each of the 20 steps of training cycles, the loss function is minimized till convergence. 
To quantify the convergence, we also evaluate the MMD distance (see Fig.~\ref{fig:MMD_cost_cluster_two_qubit}) $\calD(\tilde{\calS}_t, \calE_0)$ between the true distribution $\calE_0$ and the trained ensemble of states $\tilde{\calS}_t$ throughout the training cycles (blue), showing a convergence toward $\calD=0$. The periodic spikes show the initial increase of the MMD distance at each training cycle, due to introducing a randomly initialized PQC in a new denoising step. For reference, we also plot the evolution of the MMD distance throughout the forward-diffusion (red circles), which starts from zero at diffusion step 0 and grows toward a larger value as the diffusion step increases (from right to left). We see the training results (blue) follow closely to the diffusion results (red) as expected. In addition, the testing results (green) also agree well with the training results (blue) for QuDDPM.

As benchmarks, we consider two major quantum generative models, QuGAN and quantum direct transport (QuDT). QuDT can be regarded as the generalization of the quantum circuit Born machine~\cite{liu2018differentiable,benedetti2019generative,coyle2020born,gili2023quantum} toward quantum data. Previous works of both models only considered a single quantum state or classical distributions~\cite{hu2019quantum, zoufal2019quantum, benedetti2019generative}, and here we generalize them to adapt to the state ensemble generation task by allowing Haar random states as inputs and introducing ancilla to be measured (see Appendix~\ref{app:benchmark}). For a fair comparison, we keep the number of variational parameters of generator circuits in QuDT and QuGAN the same as QuDDPM, listed in Tabble~\ref{tab:hyperparams_benchmark}. As shown in Fig.~\ref{fig:MMD_cost_cluster_two_qubit}, QuDT and QuGAN converge to generate ensembles with a substantial MMD deviation to the true ensemble, demonstrating QuDDPM's advantage due to its unique diffusion and denoising process.

\section{Gate complexity and convergence}
Now we discuss the number of local gates required and convergence analysis for QuDDPM to solve an $n$-qubit generative task. For simplicity, we assume the qubits are one-dimensional with nearest-neighbor interactions, while similar counting can be done for other cases. To guarantee convergence toward noise, the forward scrambling circuits need a linear number of layers in $n$ as predicted by $K$ design~\cite{brandao2016local,harrow2023approximate}, leading to $O(n^2)$ total gates. 
The backward circuit will be similar, with at most $n_A\le 2n$ additional ancillas and $O(n^2)$ gates, leading the overall gate complexity of QuDDPM to be $O(n^2)$.

\begin{figure}[t]
    \centering
    \includegraphics[width=0.45\textwidth]{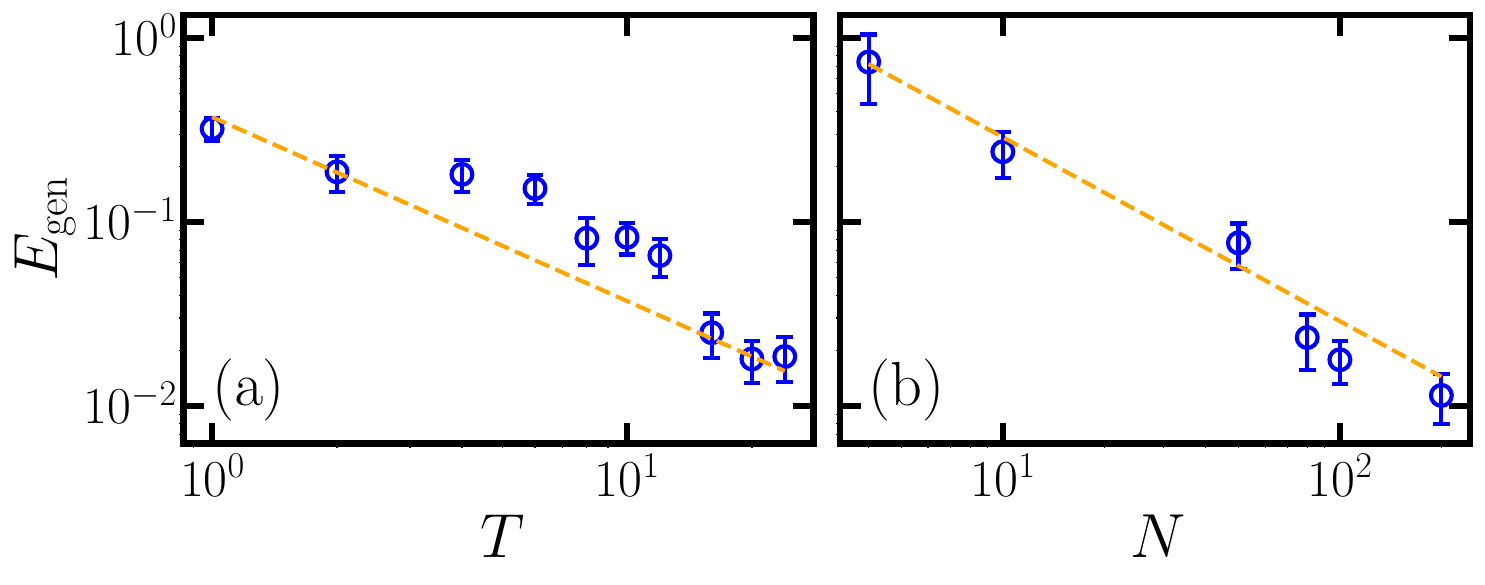}
    \caption{The generalization error of QuDDPM in generating cluster states versus (a) diffusion steps $T$ and (b) training dataset size $N$. Dots are numerical results and orange dashed line is linear fitting results with both exponents equal to $1$ within the numerical precision. }
    \label{fig:gen_error}
\end{figure}

Similar to the classical case~\cite{chen2022sampling}, the total error of QuDDPM involves three parts, 
\begin{align}
     E \simeq E_{\rm diff} + E_{\calM} + E_{\rm gen},
\end{align}
with a deviation $E_{\rm diff}$ of $\calS_T$ to true random states, measurement error $E_{\calM}$ and generalization error $E_{\rm gen}$.
We discuss the scaling of three parts separately in the following. Suppose the diffusion circuits approach an approximate $K$ design; its diffusion error is known as~\cite{brandao2016local}
\begin{align}
    E_{\rm diff} \sim 2^{n K} e^{-T/A(K) C} \sim \mathcal{O}\left(e^{-T}\right)
\end{align}
where $A(K) = \lceil \log_2(4K)\rceil^2 K^5 t^{3.1/\log(2)}$ is a polynomial of $K$ and $C$ is a constant determined by the circuit in a single step. For measurement, the standard error in estimating the fidelity $F_{ij}$ between any two states $\ket{\psi_i}, \ket{\tilde{\psi}_j}$ is ${\rm SE}(F_{ij}) = \sqrt{(1-F_{ij})/m}$, where $m$ is the number of repetitions of measurement. With $N$ data in the two sets $\calS, \tilde{\calS}$, the measurement error of estimating the $\textit{mean fidelity}$ is
\begin{align}
    E_\calM &= \frac{1}{N^2}\sqrt{\sum_{i,j=1}^N {\rm SE}\left(F_{ij}\right)^2}
    \sim \mathcal{O}\left(\frac{1}{N \sqrt{m}}\right).
\end{align}
Finally, we provide numerical evidence that the generalization error~\cite{caro2023out,banchi2021generalization} 
\begin{align}
    E_{\rm gen}(\{\tilde{U}_\ell\}_{1}^T) \equiv \calL(\{\tilde{U}_\ell\}_{1}^T,\calE_0, \tilde{\calE}_T)-\calL(\{\tilde{U}_\ell\}_{1}^T,\calS_0, \tilde{\calS}_T)
\end{align}
has the scaling $\mathcal{O}\left(\frac{1}{T N}\right)$,
as shown in Fig.~\ref{fig:gen_error} for an $n=4$ qubit clustering state generation task. Here we estimate the generalization error via a validation set independently sampled, while the proof is an open problem~\cite{caro2023out,banchi2021generalization}. The $1/N$ scaling agrees with classical machine learning results~\cite{srebro2010smoothness,yao2022mean}. 

\begin{figure}[t]
     \centering
     \includegraphics[width=0.48\textwidth]{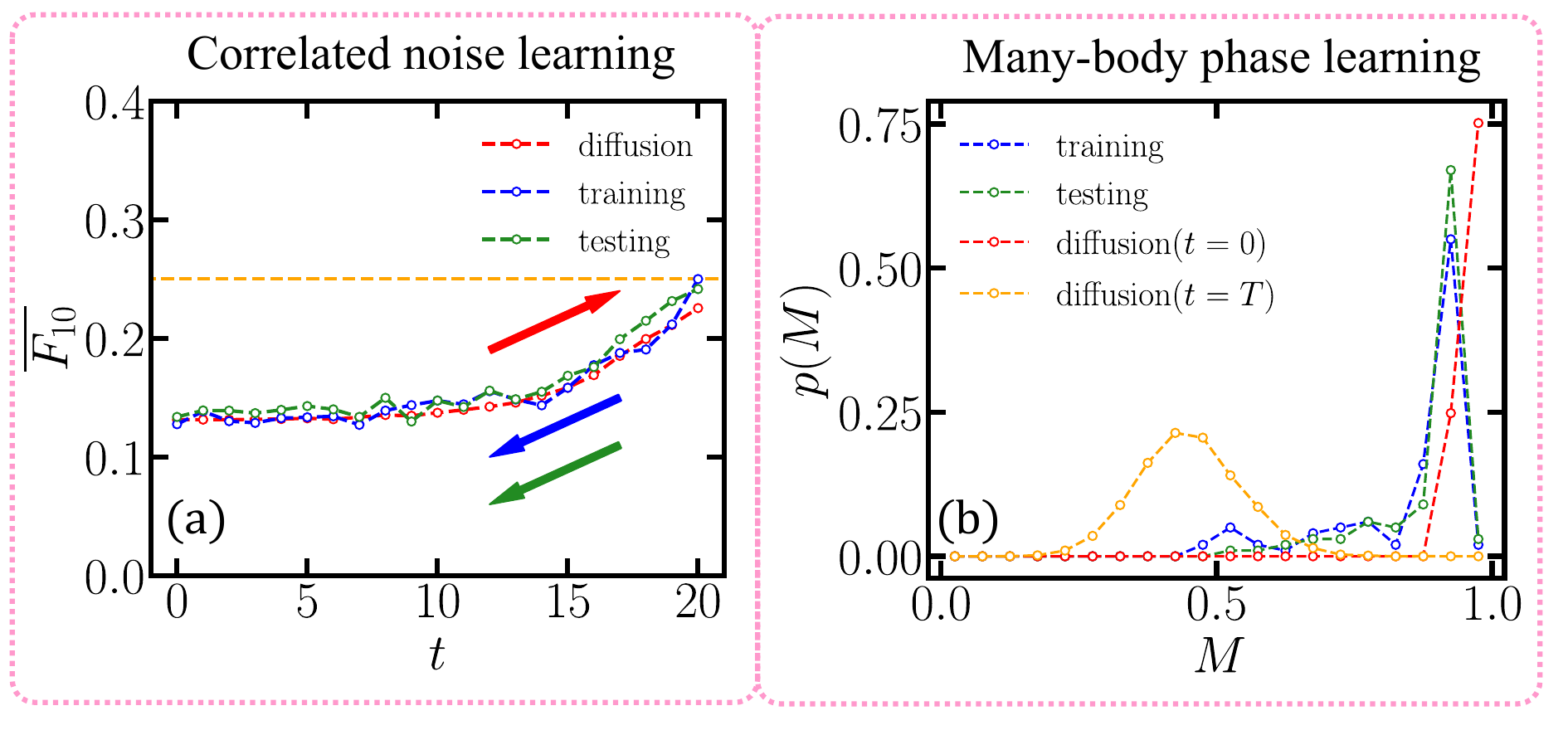}
     \caption{
     Generation of states with probabilistic correlated noise on a specific state in (a) and (b) states with ferromagnetic phase. In (a), average fidelity $\overline{F_{10}}$ between states at step $t$ and $\ket{10}$ for diffusion (red), training (blue) and testing (green) are plotted. In (b), we show the distribution of magnetization for generated data from training (blue) and testing (green) dataset, and compared to true data (red) and full noise (orange). Four qubits are considered in (b).
     }
     \label{fig:qddpm_cluster}
\end{figure}

\section{Applications}To showcase QuDDPM's applications, we consider a particular realization of QuDDPM with each unitary $U_k^{(i)}$ and $\tilde{U}_k$ implemented by the fast scrambling model~\cite{belyansky2020minimal}---layers of general single-qubit rotations in between homogeneous tunable entangling layers of all-to-all $ZZ$ rotations---and hardware efficient ansatz~\cite{kandala2017hardware}---layers of $X$ and $Y$ single-qubit rotations in between layers of nearest-neighbor control-$Z$ gates---separately (see Appendix~\ref{app:gate_set}). While the MMD distance characterization similar to Fig.~\ref{fig:MMD_cost_cluster_two_qubit} is presented in Appendix~\ref{app:distance_app}, we adopt more direct measures of performance in each application.

\subsection{Learning correlated noise}
When a real quantum device is programmed to generate a quantum state, it inevitably suffers from potentially correlated errors in the gate control parameters~\cite{harper2020efficient,chen2023learnability}. As a result, the generated states $\calS_0$ are close to the target state but have nontrivial coherent errors, which can be learned by QuDDPM.
We take a $2$-qubit example of the
target state $\ket{\Psi} = c_0\ket{00}+c_1\ket{01}+c_3\ket{11}$ under the influence of fully correlated noise, where $e^{-i\delta X_1 X_2}$ and $e^{-i\delta Z_1 Z_2}$ rotations  happen with probability $p$ and $1-p$. Here $X_k$ and $Z_k$ are Pauli operators for qubit $k$. In each case, the angle of rotation $\delta$ is uniformly sampled from the range $[-\delta_0, \delta_0]$. As the $\ket{10}$ component in the superposition only appears when $XX$ error happens, we can utilize average fidelity $\overline{F_{10}} = \E_{\tilde{S}_0} \lvert\braket{10|\psi^{(0)}}\rvert^2$ as the performance metric to estimate the error probability $p$ via $\tilde{p} = \overline{F_{10}}/\left(\lvert c_1\rvert^2  \E_\delta \left[\sin^2\delta\right]\right)$. We show a numerical example in Fig.~\ref{fig:qddpm_cluster}(a), where the generated ensemble average fidelity in training and testing, agrees with the theoretical prediction up to a finite sample size deviation.

\begin{figure}
    \centering
    \includegraphics[width=0.45\textwidth]{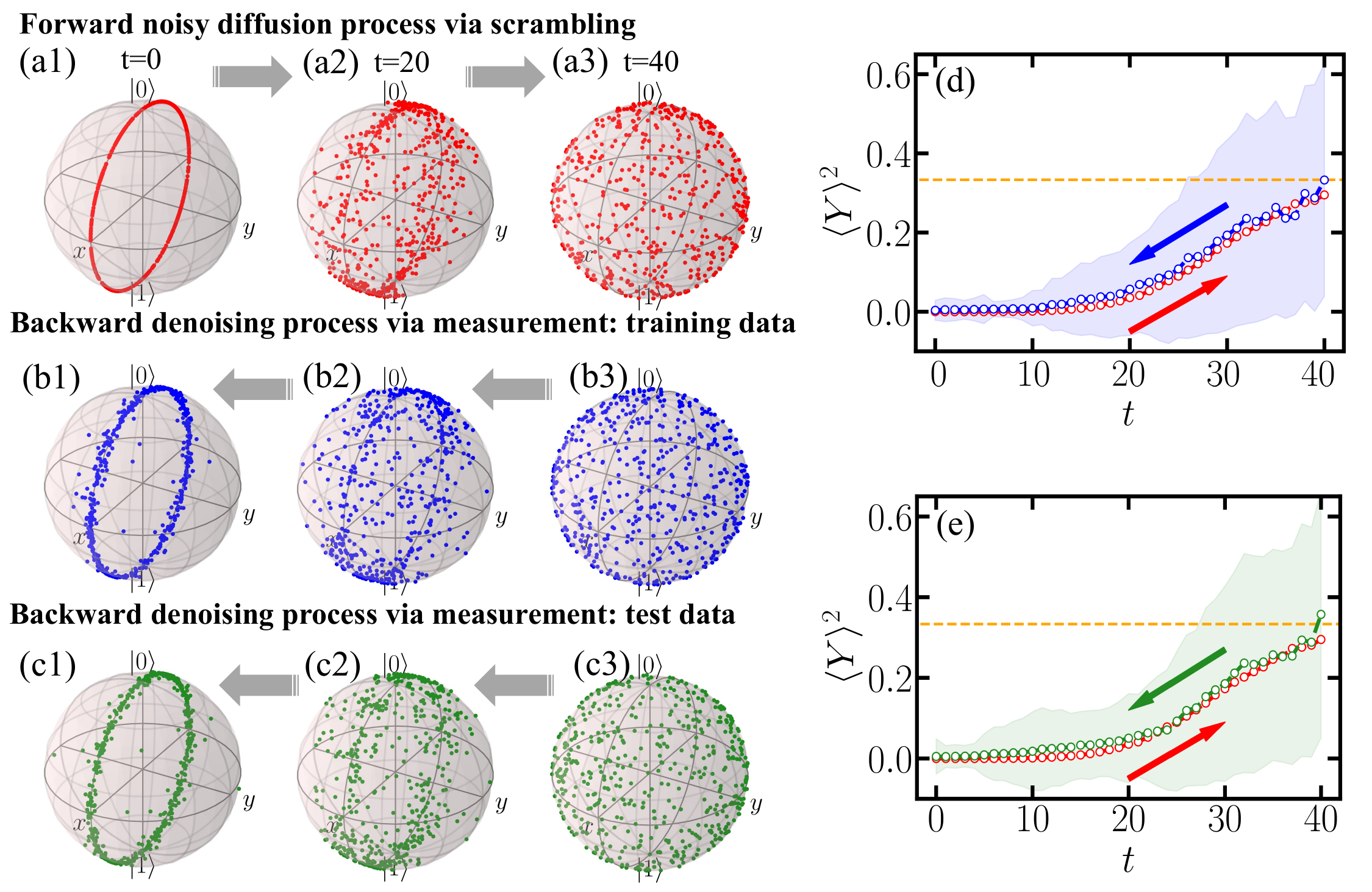}
    \caption{Bloch visualization of the forward (a1)-(a3) and backward (b1)-(b3),(c1)-(c3) process.
    (d),(e) deviation of generated states from unit circle in $X$-$Z$ plane. The deviation $\braket{Y}^2$ for forward diffusion (red), backward training (blue), and backward test (green) are plotted. The shaded area shows the sample standard deviation. 
    \label{fig:qddpm_circle}
    }
\end{figure}

\subsection{Learning many-body phases}
As a proof of principle, we take the simple and well-known transverse-field Ising model (TFIM) described by the Hamitonian
$
H_{\rm TFIM} = -\sum_i Z_i Z_{i+1} - g\sum_i X_i.
$
When $g$ increases from zero, the system undergoes a phase transition from an ordered ferromagnetic phase to a disordered phase, with the critical point at $g=1$. The states before diffusion are chosen from ground states of $H_{\rm TFIM}$ with $g\in [0.2, 0.4)$ uniformly distributed. To test the capability of QuDDPM, we utilize the magnetization,
$
M =  (\sum_i Z_i)/n,
$
to identify the phase of generated states from QuDDPM, and show its distribution in Fig.~\ref{fig:qddpm_cluster}(b). Most generated states (blue and green) of QuDDPM lives in the ferromagentic phase, and shows a sharp contrast to the random states (orange).


\subsection{Learning nontrivial topology}
We consider the ensemble of states with a ring structure---generated by applying a unitary on a single state, e.g., 
$
\ket{\psi_i } =e^{-i \bm x_i \cdot \bm G} \ket{0},
$ 
which models the scenario where one encodes the classical data $\bm x_i$ onto the quantum data $\psi_i$, as commonly adopted in quantum machine learning to solve classical problems~\cite{havlivcek2019supervised, schuld2021supervised, li2022concentration, huang2021experimental}.
We test QuDDPM with a single qubit toy example, where the generators are chosen as Pauli $Y$ and the rotation angles are uniform in $[0, 2\pi)$. In the QuDDPM training, we use the Wasserstein distance~\cite{oreshkov2009distinguishability} (see Appendix~\ref{app:wasserstein_dist}) to cope with the nontrivial topology. The forward noisy diffusion process on the sample data and the backward denoising process for training and testing are depicted in Fig.~\ref{fig:qddpm_circle}. To quantitatively evaluate the performance of QuDDPM, we evaluate the deviation by Pauli $Y$ expectation $\braket{Y}^2$ in Figs.~\ref{fig:qddpm_circle}(d) and~\ref{fig:qddpm_circle}(e), where gradual transition between zero and a Haar value of $1/3$ is observed in both forward diffusion and backward denoising.


\begin{acknowledgements}
Q.Z. and B.Z. acknowledges support from NSF CAREER Awards No. CCF-2240641 and No. ONR N00014-23-1-2296. X.C. acknowledges support from NSF CAREER Award No. DMS-2347760.
\end{acknowledgements}

\appendix

\section{On details of parameters}
We list hyperparameters and performance for all generative learning tasks in Table~\ref{tab:hyperparams} for reference, and state the targeted distribution of states to generate in the following. The major codes and data of the work can be found in Ref.~\cite{Github}. We also specify the cost function, among the two choices, maximum mean discrepancy (MMD) and Wasserstein distance.

\begin{table*}
\centering
\begin{tabular}
{|c|c|c|c|c|c|c|l|}
\hline
   Generation task  &  $n$ &  $n_A$ & $L$ &  $T$ &  $N$ & Cost function & \multicolumn{1}{c|}{Performance}\\
   \hline
   \begin{tabular}{@{}l@{}}
   Clustered state\\
   (Fig.~\ref{fig:schematic} of main text \\ 
    and Figs.~\ref{fig:MMD_cost_cluster} and \ref{fig:cluster_supp}a of Appendix) 
   \end{tabular}
   & $1$ & $1$ & $4$ & $20$ & $100$ & MMD & 
   \begin{tabular}{@{}l@{}}
   $\overline{F_{0,\rm data}}=0.987\pm 0.013$\\ 
   $\overline{F_{0,\rm tr}} = 0.992\pm 0.021$\\ 
   $\overline{F_{0,\rm te}} = 0.993 \pm 0.014$
   \end{tabular}
   \\ \hline 
   \begin{tabular}{@{}l@{}}
   Clustered state\\
   (Fig.~\ref{fig:MMD_cost_cluster_two_qubit} of main text and \\ 
    Fig.~\ref{fig:cluster_supp}b of Appendix.) 
   \end{tabular}
   & $2$ & $1$ & $6$ & $20$ & $100$ & MMD & 
   \begin{tabular}{@{}l@{}}
   $\overline{F_{0,\rm data}}=0.977 \pm 0.014$ \\
   $\overline{F_{0,\rm tr}} = 0.952\pm 0.070$ \\ 
   $\overline{F_{0,\rm te}} = 0.944 \pm 0.075$
   \end{tabular}
   \\ \hline
   Clustered state [Figs.~\ref{fig:gen_error}(a) and \ref{fig:gen_error}(b)] & $4$ & $2$ & $8$ & 20 [Fig.~\ref{fig:gen_error}(b)] & 100 [Fig.~\ref{fig:gen_error}(a)] & MMD & 
   see Figs.~\ref{fig:gen_error}\\
   \hline
   Correlated noise [Fig.~\ref{fig:qddpm_cluster}(a)] & $2$ & $2$ & $6$ & $20$ & $500$ & MMD& 
   \begin{tabular}{@{}ll@{}}
   Data: & 0.129\\
   Training: & 0.128\\ 
   Testing: & 0.133
   \end{tabular}
   \\ \hline
   Many-body phase [Fig.~\ref{fig:qddpm_cluster}(b)] & $4$ & $2$ & $12$ & $30$ & $100$ & MMD & 
   \begin{tabular}{@{}ll@{}}
   \noindent\rlap{Measured by magnetization.}\\
   Data: & 1\\
   Training: & 0.9\\
   Testing: & 0.96
   \end{tabular}
   \\ \hline
   Circular states (Fig.~\ref{fig:qddpm_circle}) & $1$ & $2$ & $6$ & $40$ & $500$ & Wasserstein & 
   \begin{tabular}{@{}ll@{}}
   Data: & $\overline{\braket{Y}^2} =0$ \\
   Training: & $\overline{\braket{Y}^2} =0.00367\pm 0.0251$\\
   Testing: &  $\overline{\braket{Y}^2} = 0.00506\pm 0.0439$
   \end{tabular}
   \\ \hline
\end{tabular}
\caption{List of hyperparameters of quantum denoising diffusion probabilistic model (QuDDPM) and its performance in different generative learning tasks. To test the performance after training, we randomly sample $N_{\rm te}$ random noise states, and perform the optimized backward PQC to generate the sampled data. Dataset size $N_{\rm tr}=N_{\rm te}=N$. $n$ is the number of data qubit $n$ and $n_A$ the ancilla qubit. $L$ is the PQC depth. $T$ is the diffusion steps. For cluster state generation, we evaluate the average fidelity with the center state in each cluster, i.e. $\ket{0}$ for single qubit and $\ket{0,0}$ for $2$ qubit (see Appendix~\ref{app:cluster}).}
\label{tab:hyperparams}
\end{table*}

\begin{table*}
\centering
\begin{tabular}
{|c|c|c|c|c|c|c|l|}
\hline
   Model  &  $n$ &  $n_A$ & No. variational parameters &  $N$ & Cost funcion & \multicolumn{1}{c|}{Performance}\\
   \hline
   QuDDPM
   & $2$ & $1$ & $720$ & $100$ & MMD & 
   \begin{tabular}{@{}l@{}}
   $\overline{F_{0,\rm tr}} = 0.947\pm 0.070$\\ 
   $\overline{F_{0,\rm te}} = 0.948\pm 0.061$
   \end{tabular}
   \\ \hline 
   QuDT
   & $2$ & $1$ & $720$ & $100$ & MMD & 
   \begin{tabular}{@{}l@{}}
   $\overline{F_{0,\rm tr}} = 0.572\pm 0.321$\\ 
   $\overline{F_{0,\rm te}} = 0.465\pm 0.349$
   \end{tabular}
   \\ \hline
   QuGAN
   & $2$ & $1$ & 
   \begin{tabular}{@{}l@{}}
   $720$ (generator)\\
   $96$ (discriminator)
    \end{tabular}
   & $100$ & Error probability based Cost function & 
   \begin{tabular}{@{}l@{}}
   $\overline{F_{0,\rm tr}} = 0.570\pm 0.250$\\ 
   $\overline{F_{0,\rm te}} = 0.443\pm 0.269$
   \end{tabular}
   \\ \hline
\end{tabular}
\caption{List of hyperparameters of quantum denoising diffusion probabilistic model (QuDDPM), quantum direct transport (QuDT) and quantum generative adversarial networks (QuGANs) for generating a clustered state in Fig.~\ref{fig:MMD_cost_cluster_two_qubit} and Fig.~\ref{fig:benchmark}. Dataset size $N_{\rm tr}=N_{\rm te}=N$. $n$ is the number of data qubit $n$, and $n_A$ is the ancilla qubit. In performance, the mean fidelity with the center state of the cluster $\ket{0,0}$ is $\overline{F_0} = \mathbb{E}_{\ket{\psi} \in \tilde{S}}|\braket{0,0|\psi}|^2$ and for true data it is $\overline{F_{0,\rm data}}=0.977 \pm 0.014$ (see Appendix~\ref{app:benchmark}).}
\label{tab:hyperparams_benchmark}
\end{table*}

In Figs.~\ref{fig:schematic}(b) and~\ref{fig:schematic}(c), we consider data in the form of $\ket{\psi^{(0)}} \sim \ket{0} + \epsilon c_1 \ket{1}$ up to a normalization constant where $\Re{c_1},\Im{c_1}\sim \calN(0,1)$ is Gaussian distributed, and the scale factor is chosen as $\epsilon = 0.08$.
We have taken single qubit rotations as $U_k^{(i)}$, where each angle is randomly sampled, e.g., from the uniform distribution $U[-\pi/8,\pi/8]$. 
In the generation of states with probabilistic correlated noise of Fig.~\ref{fig:qddpm_cluster}(a), the noise perturbation range is $\delta \in [-\pi/3, \pi/3]$.

\section{On Wasserstein distance}
\label{app:wasserstein_dist}
For pure states, choosing the quantum trace distance (equaling infidelity) $D^2(\ket{\phi}, \ket{\psi})=1-\lvert\braket{\phi|\psi}\rvert^2$, then Kantorovich's formulation for optimal transportation involves solving the following optimization problem
\begin{align}
    \mathsf{OPT} := \min_{\pi\in\Pi(\calE_1,\calE_2)}\int_{V \times V}D^p(\ket{\phi},\ket{\psi})\,d\pi(\ket{\phi}, \ket{\psi})
    \label{eq:wass_distance}
\end{align}
for $p \geq 1$, where $\Pi(\calE_1, \calE_2)$ is the set of admissible \emph{transport plans} (i.e., \emph{couplings}) of probability distributions on $V \times V$ such that $\pi(B \times V) = \calE_1(B)$ and $\pi(V \times B) = \calE_2(B)$ for any measurable $B \subset V$; namely $\Pi(\calE_1,\calE_2)$ stands for all distributions with marginals as $\calE_1$ and $\calE_2$. The Kantorovich problem in~\eqref{eq:wass_distance} induces a metric, known as the \emph{$p$-Wasserstein distance}, on the space $\calP_p(V)$ of probability distributions on $V$ with finite $p$th moment. In particular, the $p$-Wasserstein distance $W_p(\calE_1, \calE_2) = \mathsf{OPT}^{1/p}$, and it has identifiability in the sense that $W_p(\calE_1, \calE_2) = 0$ if and only if $\calE_1 = \calE_2$. More details can be found in Appendix.~\ref{app:wasserstein_compute}.

\section{On related works}
The proposed QuDDPM represents an application of the theoretical idea of quantum information scrambling~\cite{nahum2017quantum,nahum2018operator} in the forward diffusion, and its backward denoising also connects to the measurement-induced phase transitions~\cite{skinner2019measurement}. Here we point out that our forward diffusion circuits include an actual implementation of scrambling as part of the QuDDPM algorithm, while previous papers utilize tools from the study of quantum scrambling to understand quantum neural networks~\cite{shen2020,garcia2022quantifying}.

Below, we discuss about several related works. Reference~\cite{rodriguez2019identifying} utilizes a diffusion map (DM) for unsupervised learning of topological phases and~\cite{CHEN2021303} proposes a diffusion $K$-means manifold clustering approach based on the diffusion distance~\cite{COIFMAN20065}. A quantum DM algorithm has also been considered~\cite{sornsaeng2021quantum} for potentially quantum speedup. However, these works are not on generative learning and do not consider any denoising process. 
Layerwise training~\cite{skolik2021layerwise} also attempts to divide a training problem into subtasks in nongenerative learning; however, the performance of such strategies is limited~\cite{campos2021training}. QuDDPM integrates the division of training task and an actual noisy diffusion process to enable provable benefits in training. After the completion of our work, we became aware of a recent paper~\cite{parigi2023quantum}, where hybrid quantum-classical DDPM is proposed. Our work focuses on quantum data; provides explicit construction of quantum diffusion and quantum denoising, loss function, training strategy and error analyses; and presents several applications.

\section{On future directions}
Finally, we point out some future directions, besides various applications of QuDDPM in learning quantum systems. Our current QuDDPM architecture requires a loss function based on fidelity estimations. For large systems, fidelity estimation can be challenging to implement. Toward efficient training in large systems, alternative loss functions can be adopted. For example, one may consider adopting another quantum circuit trained for telling the ensembles apart, such as a quantum convolutional neural network~\cite{cong2019quantum} and other circuit architecture~\cite{zhang2022fast}. Such an approach will combine QuDDPM and the adversarial agent in QuGAN to resolve the training problem in QuGAN. Another future direction is controlled diffusion~\cite{song2021solving}: when the ensemble has special symmetry, one can restrict the forward scrambling, the backward denoising and the random noise ensemble to that symmetry. It is also an interesting open problem of how to introduce a control knob such that QuDDPM can learn multiple distributions and generate states according to an input requesting one of the distributions.

Besides learning quantum errors and many-body phases, quantum sensor networks~\cite{zhuang2019physical,xia2021quantum} provide another application scenario of QuDDPM. In this scenario, one sends quantum probes to sense a unitary physical process; on the return side, the receiver will collect a pure state from a distribution in the ideal case. It is an open problem of how QuDDPM can be adopted to provide an advantage in quantum sensing.

Furthermore, the data can also be quantum states encoding classical data, where QuDDPM can also process classical data. Benchmarking QuDDPM and previous algorithms for classical data generative learning is an open direction.

\section{Review of Classical DDPM} 
\label{app:review}

In classical DDPM, forward diffusion process first gradually converts the observed data to a simple random noise based on non-equilibrium thermodynamics, and then an associated reverse-time process is learned to generate samples with target distribution from the noise~\cite{sohl2015deep,ho2020denoising,song2021scorebased,schneuing2022structurebased}

Classical DDPM can be viewed as a latent variational autoencoder (VAE) model with stochastic hidden layers of the same dimension as the input data. 
The forward diffusion process simply adds a sequence of small amount of Gaussian perturbations to the data sample $x_0$ in $T$ steps to produce the noisy samples $x_1, \dots, x_T$ according to a linear Markov chain:
\begin{equation}
\label{eqn:ddpm_forward}
q(x_t \mid x_{t-1}) = N\bigl(x_t \mid \sqrt{1-\beta_t} x_{t-1}, \beta_t I\bigr),
\end{equation}
where $\beta_t \in (0, 1)$ is a given noise schedule such that $q(x_t)$ converges to $N(0, I)$. Usually the noise schedule satisfies $\beta_1 < \beta_2 < \cdots < \beta_T$ so that larger step sizes are used when the samples become more noisy. In the reverse-time process, we would like to sample from $q(x_{t-1} \mid x_t)$, which allows us to generate new data sample from the noise distribution $q(x_T)$. However, the conditional distribution $q(x_{t-1} \mid x_t)$ is intractable and approximated by a decoder of the form:
\begin{equation}
\label{eqn:ddpm_backward}
p_{\theta}(x_{t-1} \mid x_t) = N\bigl(x_{t-1} \mid \mu_{\theta}(x_t, t), \sigma^2_t I\bigr),
\end{equation}
where the time-dependent conditional mean vector $\mu_{\theta}$ is parameterized by a neural network. Then the training of $\mu_{\theta}$ can be efficiently performed by maximizing an evidence lower bound (ELBO) for the log-likelihood function $\log p_{\theta}(x_0)$. Refs.~\cite{sohl2015deep,ho2020denoising} showed that the ELBO can be expressed as a linear combination of (relative) entropy terms for Gaussian distributions that can be evaluated analytically into a simple weighted $L^2$ loss function.

In Ref.~\cite{song2021scorebased}, it is shown that DDPM forward process is the discretized version of the following continuous stochastic differential equation (SDE): for $t \in [0, T]$,
\begin{equation}
    \label{eqn:ddpm_continuous_forward}
    \mathrm{d} x(t) = -{1\over2} \beta(t) x(t) \; \mathrm{d} t + \sqrt{\beta(t)} \; \mathrm{d} w(t),
\end{equation}
where $w(t)$ is the standard Brownian motion, and the DDPM decoder~\eqref{eqn:ddpm_backward} corresponds to the discretization of a reverse-time SDE given by
\begin{equation}
    \label{eqn:ddpm_continuous_backward}
    \begin{gathered}
        \mathrm{d} x^{\leftarrow}(t) = -{1\over2}\beta(t) [ x^{\leftarrow}(t) + 2 \nabla_x \log p_t(x^\leftarrow(t)) ] \; \mathrm{d} t \\
        + \sqrt{\beta(t)} \; \mathrm{d} w^{\leftarrow}(t),
    \end{gathered}
\end{equation}
where $w^{\leftarrow}(t)$ is the standard Brownian motion flowing back in time and $p_t(\cdot)$ is the probability density of $x(t)$. The forward process $x(t)$ in~\eqref{eqn:ddpm_continuous_forward} and the reverse-time process $x^\leftarrow(T-t)$ in~\eqref{eqn:ddpm_continuous_backward} have the same marginal probability densities~\cite{ANDERSON1982313}. Thus in the forward flow, estimation of the conditional mean vector $\mu_{\theta}(\cdot,t)$ in DDPM is equivalent to learn the time-dependent score $\nabla_x \log p_t(\cdot)$ that contains the full information of data distribution $p_0$. Such score estimate is subsequently used in the reverse time process for sampling new data from $p_0$.

\section{Reproducing kernel Hilbert spaces and maximum mean discrepancy}
\label{app:RKHS_MMD} 
A bivariate function $F : V \times V \to \mathbb{R}$ is a positive definite kernel if $\sum_{i,j=1}^m c_i c_j F(\ket{\phi}_i, \ket{\phi}_j) \geq 0$ for all $m \geq 1, \ket{\phi}_1, \dots, \ket{\phi}_m \in V$, and $c_1, \dots, c_m \in \mathbb{R}$. From the Moore-Aronszajn theorem (see e.g.,~\cite[Theorem 7.2.4]{HsingEubank2015_Wiley}), for every symmetric and positive-definite kernel $F$, there is a unique Hilbert space $\calH := \calH(F)$ of real-valued functions on $V$ such that:
\begin{enumerate}[(i)]
    \item $F(\cdot, \ket{\phi}) \in \calH$ for each $\ket{\phi} \in V$;
    \item $g(\ket{\phi}) = \langle g, F(\cdot, \ket{\phi}) \rangle_{\calH}$ for each $g \in \calH$ and $\ket{\phi} \in V$.
\end{enumerate}
The space $\calH$ of functions $\{g : V \to \mathbb{R}\}$ is called the reproducing kernel Hilbert space (RKHS) associated with the kernel $F$. Property (i) defines a feature map (a.k.a. RKHS map) $V \to \calH$ via $\ket{\phi} \mapsto F(\cdot, \ket{\phi})$, and property (ii) is the reproducing kernel property for the evaluation functionals. In addition, we have for all $\ket{\phi}, \ket{\psi} \in V$,
\begin{equation*}
    F(\ket{\phi}, \ket{\psi}) = \langle F(\cdot, \ket{\psi}), \; F(\cdot, \ket{\phi}) \rangle_{\calH}=\lvert\braket{\phi|\psi}\rvert^2.
\end{equation*}
Based on the kernel $F$, the (squared) maximum mean discrepancy (MMD) loss between two state distributions is 
\begin{align}
    &\calD_\text{MMD}(\calE_1, \calE_2) 
    \nonumber
    \\
    &:= \sup_{g \in \calB}\bigl\lvert\langle \E_{\ket{\phi}\sim\calE_1}[F(\cdot, \ket{\phi})] - \E_{\ket{\psi}\sim\calE_2}[F(\cdot, \ket{\psi})],\; g \rangle_{\calH}\bigr\rvert^2,
    \label{MMD_def}
\end{align}
where $\calB$ is the unit ball in $\calH$ centered at the origin and $\langle,\rangle_{\calH}$ denotes the inner product in the RKHS.

By duality of $\calH$, we have
\begin{align*}
    &\calD_\text{MMD}(\calE_1, \calE_2)\\
    &\quad=\norm{\E_{\ket{\phi}\sim\calE_1} [F(\cdot, \ket{\phi})] - \E_{\ket{\psi}\sim\calE_2}[F(\cdot, \ket{\psi})]}_{\calH}^2.
\end{align*}
Therefore, we may calculate the MMD loss as following
\begin{align*}
    &\calD_\text{MMD}(\calE_1, \calE_2)\\
    &\quad= \langle \E_{\ket{\phi}\sim\calE_1} [F(\cdot, \ket{\phi})], \; \E_{\ket{\phi'}\sim\calE_1} [F(\cdot, \ket{\phi'})] \rangle_{\calH} \\
    &\qquad + \langle \E_{\ket{\psi}\sim\calE_2}[F(\cdot, \ket{\psi})], \; \E_{\ket{\psi'}\sim\calE_2}[F(\cdot, \ket{\psi'})] \rangle_{\calH} \\
    &\qquad - 2 \langle \E_{\ket{\phi}\sim\calE_1} [F(\cdot, \ket{\phi})], \; \E_{\ket{\psi}\sim\calE_2}[F(\cdot, \ket{\psi})] \rangle_{\calH} \\
    &\quad= \E_{\ket{\phi}, \ket{\phi'} \sim \calE_1} [F(\ket{\phi}, \ket{\phi'}]) \\
    &\qquad + \E_{\ket{\psi}, \ket{\psi'} \sim \calE_2} [F(\ket{\psi}, \ket{\psi'})] \\
    &\qquad - 2 \E_{\ket{\phi} \sim \calE_1, \ket{\psi} \sim \calE_2} [F(\ket{\phi}, \ket{\psi})],
\end{align*}
where all random states $\ket{\phi}, \ket{\phi'}, \ket{\psi}, \ket{\psi'}$ are drawn independently. Hence, we have established the connection
\begin{equation*}
    \calD_\text{MMD}(\calE_1, \calE_2) = \overline{F}(\calE_1, \calE_1) + \overline{F}(\calE_2, \calE_2)-2\overline{F}(\calE_1, \calE_2),
\end{equation*}
where $\overline{F}(\calE_1, \calE_2) := \E_{\ket{\phi} \sim \calE_1, \ket{\psi} \sim \calE_2} [F(\ket{\phi}, \ket{\psi})]$. In particular, if $F = \lvert\braket{\phi|\psi}\rvert^2$ is state-wise fidelity, then the resulting MMD corresponds to the mean fidelity defined in Eq.~\eqref{eqn:fidelity} of the main text.

\section{Computation of Wasserstein distance}
\label{app:wasserstein_compute}
In the discrete or empirical cases, where \(\calS_1\) and \(\calS_2\) are supported over finite number of state vectors \(\{\ket{\phi_i}\}_{i=1}^m\) and \(\{\ket{\psi_j}\}_{j=1}^n\), computation of Wasserstein distance can be cast into a linear program \cite{peyr2020computational}. To this end, let \(a\) and \(b\) be histograms representing \(\calS_1\) and \(\calS_2\), respectively. Define the \(m\times n\) cost matrix \(C\) by \(C_{i,j}\coloneqq D^2(\ket{\phi_i}, \ket{\psi_j})\), where $\ket{\phi_i}$ and $\ket{\psi_j}$ are state vectors from sampled ensemble state $\calS_1$ and $\calS_2$, respectively. For pure states, the cost matrix reduces to a function of infidelity, i.e., $C_{i,j}=1-F^2(\ket{\phi_i}, \ket{\psi_j})$. Then we have
\begin{align*}
W_2(\calS_1,\calS_2)=\min_{P}\quad&\langle P, C\rangle,\\
\text{s.t.}\quad&P\bm{1}_n=a,\\
\quad&P^\top\bm{1}_m=b,\\
\quad&P\geq0.
\end{align*}

Furthermore, we note that the above generalization of Wasserstein distance to characterize separation between two ensembles of quantum states is different from the case of Refs.~\cite{chakrabarti2019quantum,de2021quantum,kiani2022learning}, where only a pair of quantum states are considered.

\section{SWAP test}
\label{app:swap}
In the QuDDPM framework in the main text, we need to evaluate the fidelities between states $\ket{\psi}$ from real diffusion ensemble $\calS_k$ and the ones $\ket{\tilde{\psi}}$ from backward generated ensemble $\tilde{\calS}_k$. For any two pure states $\ket{\psi}$ and $\ket{\tilde{\psi}}$, one can perform the SWAP test to obtain the fidelity, which is illustrated in Fig.~\ref{fig:swap_test}. The SWAP test consists of two Hadamard gate and a controlled-swap gate applied on $2n+1$ qubits. Given the input $\ket{0,\psi,\tilde{\psi}}$, the output state ahead of measurement is
\begin{equation}
\begin{aligned}
\ket{0,\psi,\tilde{\psi}}&\rightarrow \frac{1}{2}\ket{0}\left(\ket{\psi,\tilde{\psi}}+ \ket{\psi,\tilde{\psi}}\right)\\
&\quad+ \frac{1}{2}\ket{1}\left(\ket{\psi,\tilde{\psi}} - \ket{\psi,\tilde{\psi}}\right),
\end{aligned}
\end{equation}

then the probability of measure $0$ on the first qubit is 
\begin{equation}
\begin{aligned}
    &P({\rm first \: qubit\:in}\ket{0})\\
    &\quad= \frac{1}{4}\left(\bra{\psi,\tilde{\psi}} + \bra{\psi,\tilde{\psi}}\right)\left(\ket{\psi,\tilde{\psi}} + \ket{\psi,\tilde{\psi}}\right) \\
    &\quad= \frac{1}{2} + \frac{1}{2}\lvert\braket{\psi|\tilde{\psi}}\rvert^2
\end{aligned}
\end{equation}
which directly indicates the fidelity.
\begin{figure}
    \centering
    \includegraphics[width=0.4\textwidth]{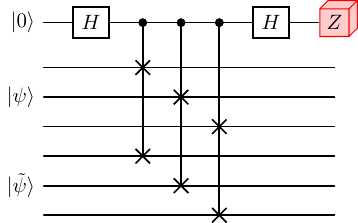}
    \caption{Circuit implementation of swap test. We show an example of swap test between two $3$-qubit state $\ket{\psi}, \ket{\tilde{\psi}}$. A $Z$-basis measurement is performed at the end.}
    \label{fig:swap_test}
\end{figure}
To implement the swap test between two $n$-qubit state $\ket{\psi}$ and $\ket{\tilde{\psi}}$, in general we need $n$ Fredkin gate between every pair of qubits of $\ket{\psi}$ and $\ket{\tilde{\psi}}$, where it is sufficient to implement a Fredkin gate among a control qubit and two target qubits via $5$ two-qubit gate~\cite{smolin1996five}. Therefore, in general $\mathcal{O}(5n)$ two-qubit gates are enough to perform the swap test in Fig.~\ref{fig:swap_test} regardless the locality of these gates.


\section{Forward and backward circuits}
\label{app:gate_set}
The froward noisy process is implemented by the fast scrambling model~\cite{belyansky2020minimal} with controllable parameters on $n$ data qubits to mimic the diffusion process in the Hilbert space. The diffusion circuit is implemented as (see Fig.~\ref{fig:circuit} (a))
\begin{align}
    U(\bm \phi_t, g_t) = \prod_{t=1}^T W_t(g_t) V_t(\bm \phi_t),
\end{align}
where $V_t$ consists of general single qubit rotations on each qubit as
\begin{equation}
    V_t(\bm \phi_t) = \bigotimes_{k=1}^{n} e^{-i\phi_{t,3k+2}Z_k/2}e^{-i\phi_{t,3k+1}Y_k/2}e^{-i\phi_{t,3k}Z_k/2},
\end{equation}
and homogeneous entangling layer $W_t$ consists of ZZ rotation on every pair of qubits, 
\begin{equation}
\begin{aligned}
    W_t(g_t) &= \exp\left[-i\frac{g_t}{2\sqrt{n}}\sum_{k_1<k_2}Z_{k_1}Z_{k_2}\right] \\
    &= \prod_{k_1<k_2} \exp\left[-i\frac{g_t}{2\sqrt{n}}Z_{k_1}Z_{k_2}\right].
\end{aligned}
\end{equation}
With a tunable range of random rotation angles $\phi$ and $g$, we can control the diffusion speed of original quantum states ensemble $\calS_0$ in the Hilbert space towards Haar random states ensemble.

The backward denoising process consists of $T$ steps, where the operation at every step is implemented by a PQC $\tilde{U}_t (\bm \theta_t)$ followed by measurements on ancillae.
\begin{equation}
    \begin{gathered}
        \Phi_t \left(\ket{\tilde{\psi}^{(t)}}\right) = \frac{\left(\Pi_A \otimes \bI_D\right) \tilde{U}_t \ket{\tilde{\psi}^{(t)}} }{\sqrt{\braket{\tilde{\psi}^{(t)}|\tilde{U}_t^\dagger\left(\Pi_A \otimes \bI_D\right) \tilde{U}_t|\tilde{\psi}^{(t)}}}} \\
     = \ket{\bm z(t)}_A\otimes \ket{\tilde{\psi}^{(t-1)}}
    \end{gathered}
\end{equation}
where $\Pi_A = \ketbra{\bm z(t)}{\bm z(t)}_A$ is the POVM of measurement on ancillas in computational basis $\ket{\bm z(t)}_A$. Note that we do not make any specific constraint on the measurement results $\bm z(t)$, instead, we simply perform the measurement on ancillas and collect post-measurement state on data qubits $\ket{\tilde{\psi}^{(t-1)}}$. 

In general, the backward PQC can utilize any architecture as long as its expressivity can guarantee for the backward transport from ensemble $\calS_{t+1}$ to $\calS_t$. In this work, we adopt the hard-efficient ansatz~\cite{kandala2017hardware} which is universal and easy to implement in practical experiments. For a $L$-layer backward PQC $\tilde{U}_t$, in each layer it consists of single qubits rotations along X and Y axes on each qubit, followed by control-Z gate on nearest neighbors as (see Fig.~\ref{fig:circuit} (b)) 
\begin{align}
    \tilde{U}_t(\bm \theta_t) = \prod_{\ell=1}^L \tilde{W}_t \tilde{V}_t(\bm \theta_t)
\end{align}
where
\begin{align}
    \tilde{V}_t(\bm \theta_t) &= \bigotimes_{k=1}^{n} e^{-i\theta_{t,2k+1}Y_k/2}e^{-i\theta_{t,2k}X_k/2}\\
    \tilde{W}_t &= \bigotimes_{k=1}^{\lfloor  (n-1)/2\rfloor} {\rm CZ}_{2k,2k+1}\bigotimes_{k=1}^{\lfloor n/2\rfloor} {\rm CZ}_{2k-1,2k}
\end{align}
with ${\rm CZ}_{k_1,k_2}$ to be the control-Z gate on qubit $k_1$ and $k_2$. 
%
%
\begin{figure}[htbp]
    \centering
    \includegraphics[width=0.45\textwidth]{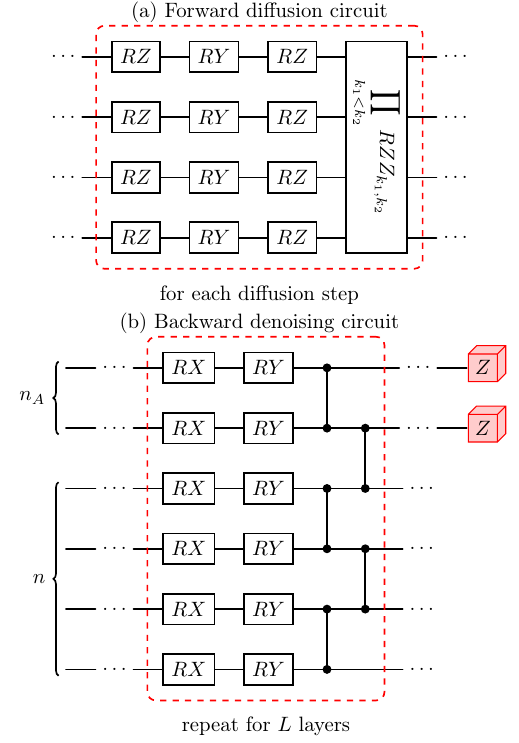}
    
    
    \caption{Quantum circuit architectures. In (a) we show the forward diffusion circuit for one diffusion step on a system of $n=4$ qubits. In (b) we show one-layer architecture of a $L$-layer backward denoising PQC on a system of $n=4$ data and $n_A=2$ ancilla qubits. RX, RY and RZ are Pauli X, Y and Z rotations. RZZ is the two-qubit ZZ rotation.}
    \label{fig:circuit}
\end{figure}

The whole backward denoising process can thus be represented as
\begin{align}
    \Phi = \Phi_1 \circ \Phi_2\circ \cdots \circ\Phi_{T-1}\circ \Phi_{T}.
\end{align}

\section{Additional details on distance metrics evaluation}
\label{app:distance_degenerate}

In Fig.~\ref{fig:distance}, we show a numerical comparison between the MMD distance (see Eq.~\ref{eq:nat_distance}) and Wasserstein distance (see Eq.~\ref{eq:wass_distance}) in different generation tasks. In the clustered state generation (Fig.~\ref{fig:distance}(a), (b)), both distance measure behaves similarly. However, for the circular state generation in Fig.~\ref{fig:distance}(c), (d) , the Wasserstein distance can characterize the diffusion of ensemble while the MMD distance fails. Note that only relative shift of MMD or Wasserstein loss through diffusion matters, while a comparison between their magnitude is unfair. In the following, we provide the simple proof on Lemma~\ref{lem:UnifCir_Haar}.

\begin{figure}[bp]
    \centering
    \vspace{1em}
    \includegraphics[width=0.45\textwidth]{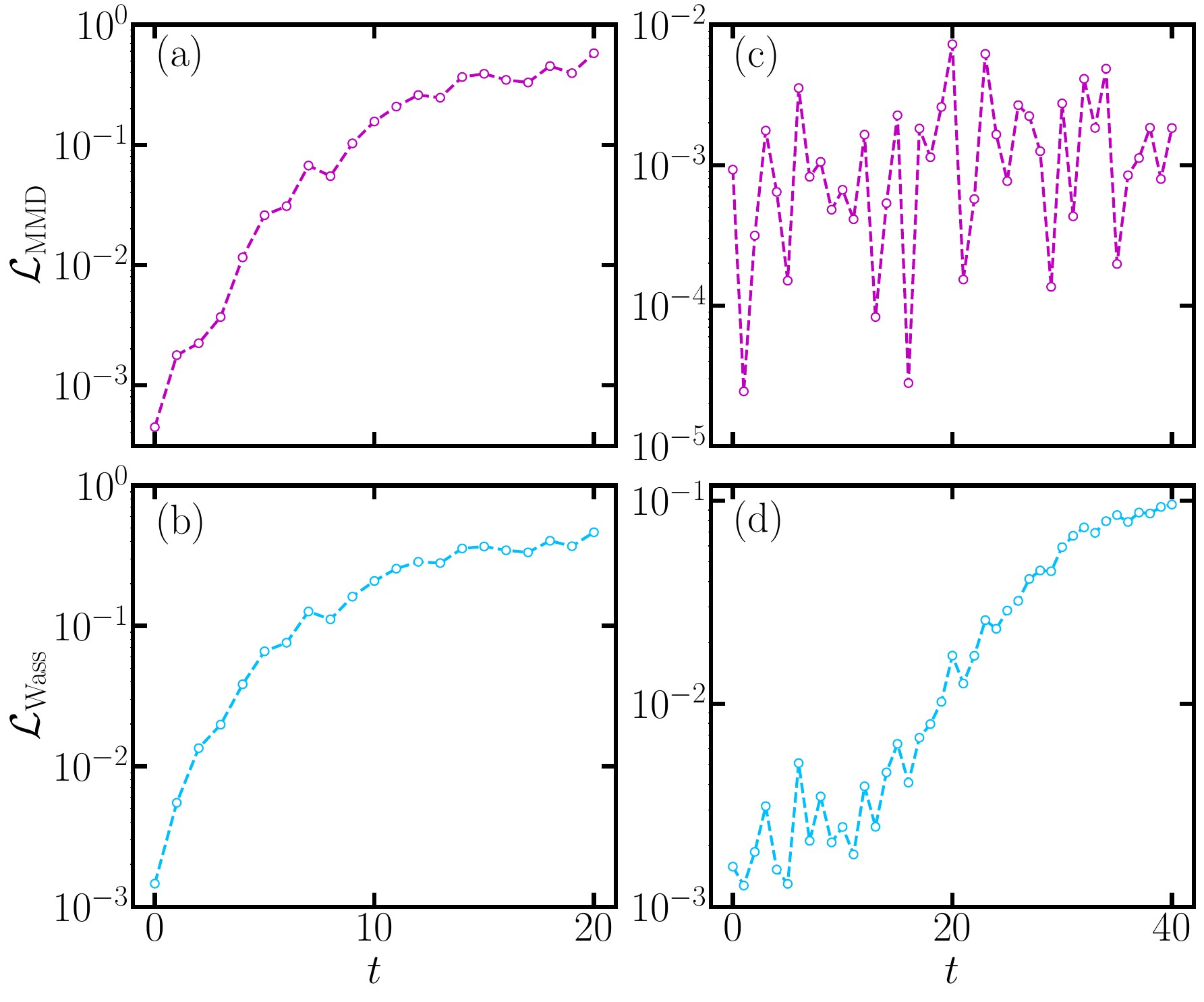}
    \vspace{-1em}
    \caption{MMD and Wasserstein loss between ensemble $\calS_0$ and ensemble through diffusion process at step $t$, $\calS_t$ in generation of cluster state (left) and circular state ensemble (right). For cluster state ensemble, MMD (a) and Wasserstein loss (c) behaves similarly. For circular state ensemble, MMD loss (c) vanishes while Wasserstein loss characterize the diffusion of distribution. 
    The data set size for cluster states is $|\calS|= 100$, and $|\calS|=500$ for circular states.}
    \label{fig:distance}
\end{figure}

\begin{lemma}
    \label{lem:UnifCir_Haar}
    Let $\calE_1$ be a uniform circular distribution on the Block sphere and $\calE_2$ be the Haar random state distribution. Then $\calD_{\rm MMD}(\calE_1, \calE_2) = 0$. 
\end{lemma}

\begin{proof}[Proof of Lemma~\ref{lem:UnifCir_Haar}]
    We can calculate the ensemble average fidelity by Haar integral as
    \begin{align*}
    \overline{F}(\calE_1, \calE_2) &= \frac{1}{2\pi}\int_{0}^{2\pi} d\theta \int d U |\braket{0|U^\dagger e^{-i\theta Y}|0}|^2 \\
    &= \frac{1}{2\pi}\int d\theta \frac{\tr(\ketbra{0}{0})\tr(e^{-i\theta Y}\ketbra{0}{0}e^{i\theta Y})}{2}\\
    &= \frac{1}{2},\\
    \overline{F}(\calE_1, \calE_1) &=  \frac{1}{(2\pi)^2}\int_{0}^{2\pi} d\theta d\theta^{\prime} |\braket{0|e^{i\theta^\prime Y} e^{-i\theta Y}|0}|^2\\
    &= \frac{1}{(2\pi)^2}\int_{0}^{2\pi} d\theta d\theta^{\prime}\cos^2\left(\theta^\prime-\theta\right)\\
    &= \frac{1}{2},\\
    \overline{F}(\calE_2, \calE_2) &=  \int d U dU^\prime |\braket{0|U^\dagger U^\prime|0}|^2\\
    &= \int dU \frac{\tr(\ketbra{0}{0})\tr(U\ketbra{0}{0}U^\dagger)}{2} \\
    &= \frac{1}{2}.
\end{align*}
Therefore, the MMD distance is
\begin{equation}
    \begin{aligned}
    \calD_{\rm MMD}(\calE_1, \calE_2)&= \overline{F}(\calE_1, \calE_1) + \overline{F}(\calE_2, \calE_2) -2\overline{F}(\calE_1, \calE_2)= 0,
    \end{aligned}
\end{equation}
which indicates its incapability to discriminate the difference between the circular state ensemble and Haar random states.
\end{proof}

For the true distribution, when $\calE_1 = \calE_2$, we should have $\calD_{\rm MMD}(\calE_1, \calE_2) = \calD_{\rm Wass}(\calE_1, \calE_2) = 0$ in theory, while in practice due to finite samples, both of them cannot vanish exactly (see Fig.~\ref{fig:distance}) left with a relatively small residual.

\section{Details of simulation}

In this work, the simulation of QuDDPM is implemented with the Python library \texttt{TensorCircuit}~\cite{zhang2023tensorcircuit} and Bloch sphere visualizations are plotted with the help from \texttt{QuTip}~\cite{johansson2012qutip}. The computation of the Wasserstein distance, on the other hand, is performed by the Python library \texttt{POT} \cite{flamary2021pot}. The major codes and data of the work are available in Ref.~\cite{Github}.

\begin{figure}
    \centering
    \includegraphics[width=0.45\textwidth]{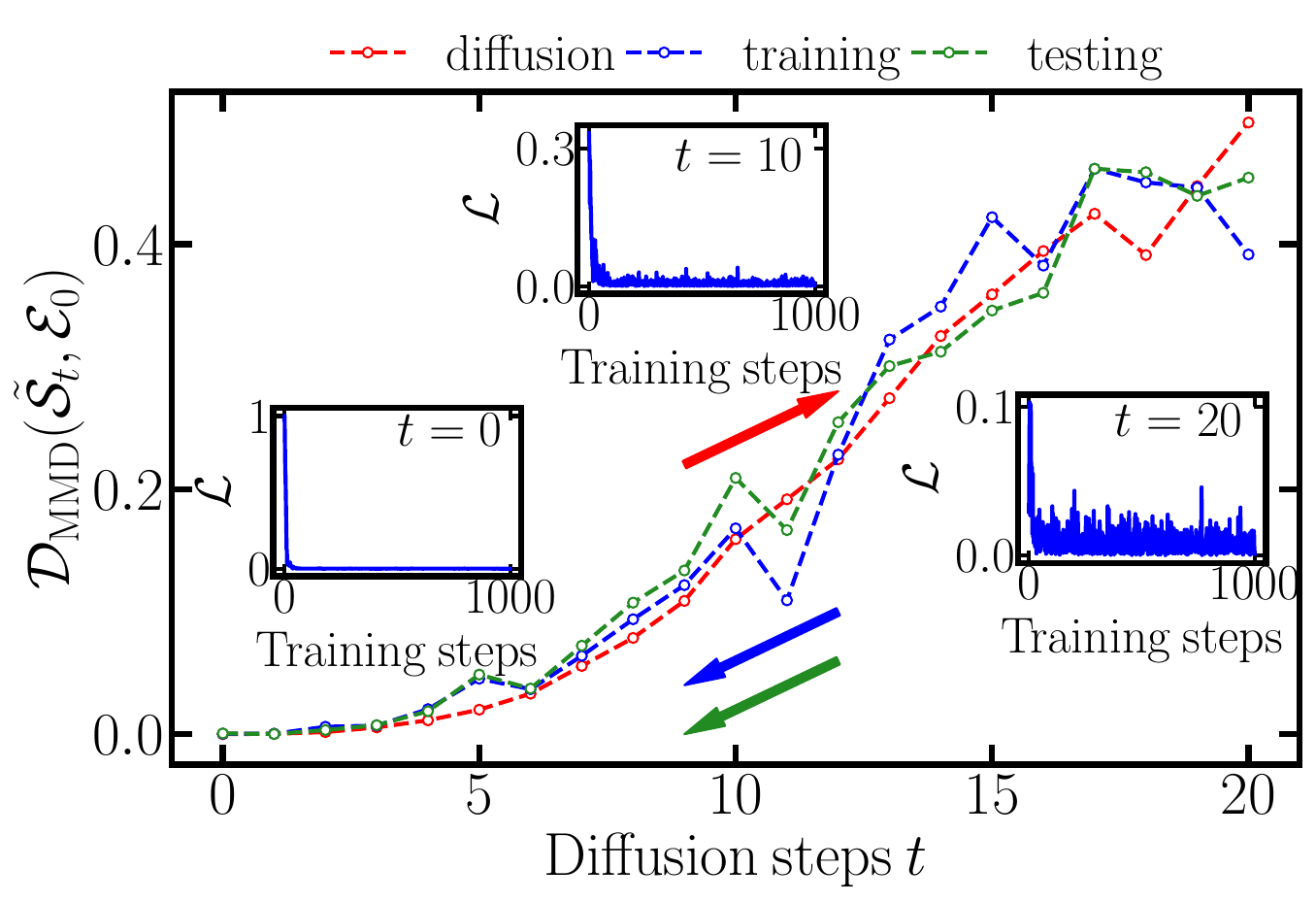}
    \caption{The decay of MMD distance $\calD$ between generated ensemble $\tilde{\calS}_t$ and target ensemble of states $\calE_0$ versus diffusion steps. The target is clustering ensemble around $\ket{0}$. Insets show the training loss history at different $t$. The converged distance at $t=0$ is less than $10^{-3}$. }
    \label{fig:MMD_cost_cluster}
\end{figure}

\section{Details of clustered state generation}
\label{app:cluster}

Here we provide more details in QuDDPM's performance in the task of clustered state generation. The single-qubit toy example in Fig.~\ref{fig:schematic} of the main text adopted the MMD distance in the loss function to generate states clustered around $\ket{0}$. In each of the 20 steps of training cycles, the loss function is minimized till convergence, as shown in the insets of Fig.~\ref{fig:MMD_cost_cluster} for $t=20,10,0$. To quantify the convergence, we also evaluate the MMD distance $\calD(\tilde{\calS}_t, \calE_0)$ between the true distribution $\calE_0$ and the trained ensemble of states $\tilde{\calS}_t$ throughout the training cycles (blue). As a reference, we also evaluate the MMD distance $\calD(\calS_t, \calE_0)$ between the ensembles $\calS_t$ in the forward diffusion steps and the true distribution $\calE_0$ throughout the diffusion steps (red). We a good agreement between the diffusion trajectory (red) and training trajectory (blue), showing a good convergence in each of the training cycles. This is consistent with the convergence of loss function as shown in the insets of Fig.~\ref{fig:MMD_cost_cluster}. At the same time, we also generate test data to verify the performance (green), showing a good agreement with the training results.

\begin{figure}
    \centering
    \includegraphics[width=0.5\textwidth]{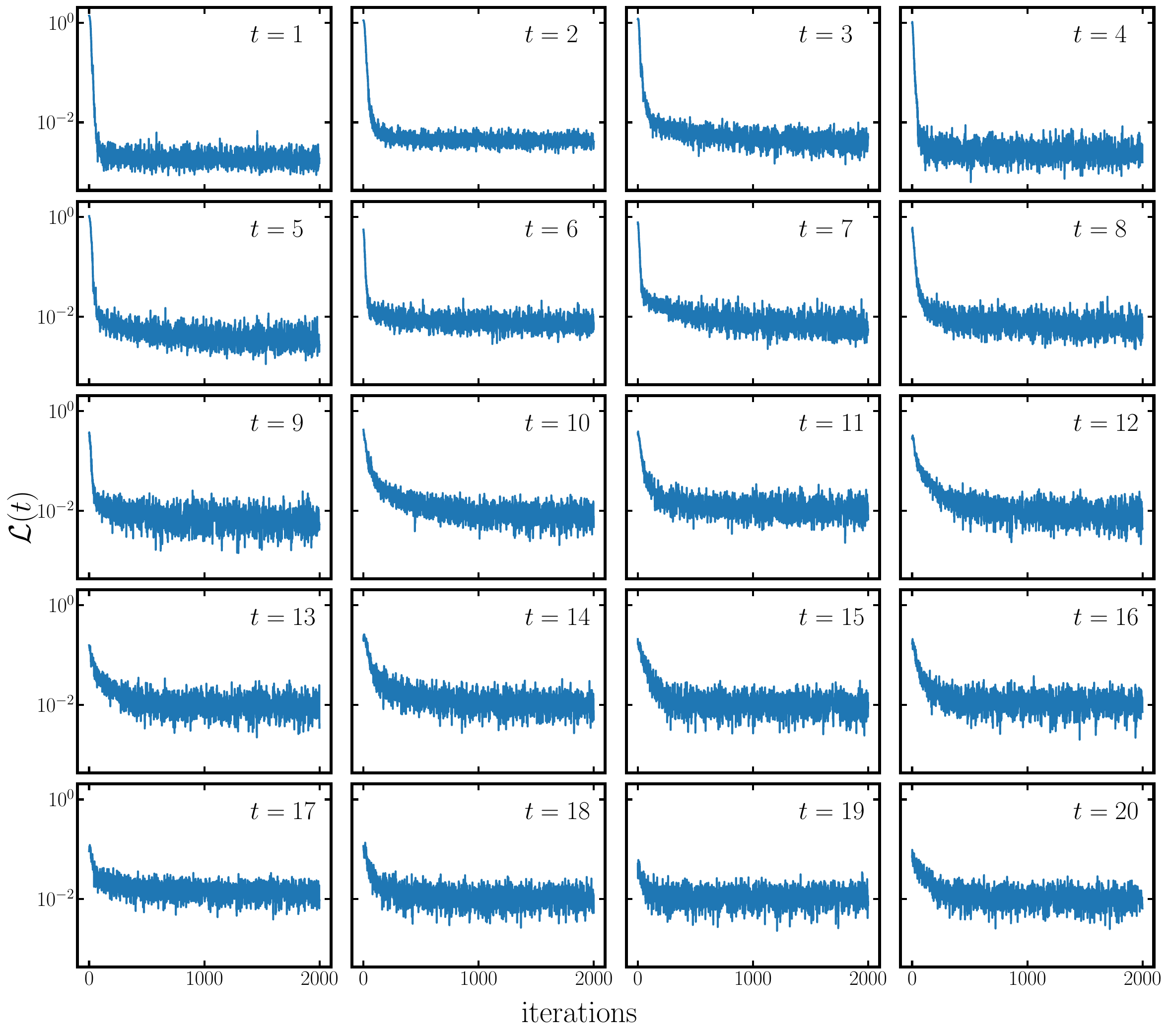}
    \caption{Training history of loss function in each denoising step for generating cluster state ensemble close to $\ket{0,0}$. }
    \label{fig:loss}
\end{figure}

\begin{figure}
    \centering
    \includegraphics[width=0.45\textwidth]{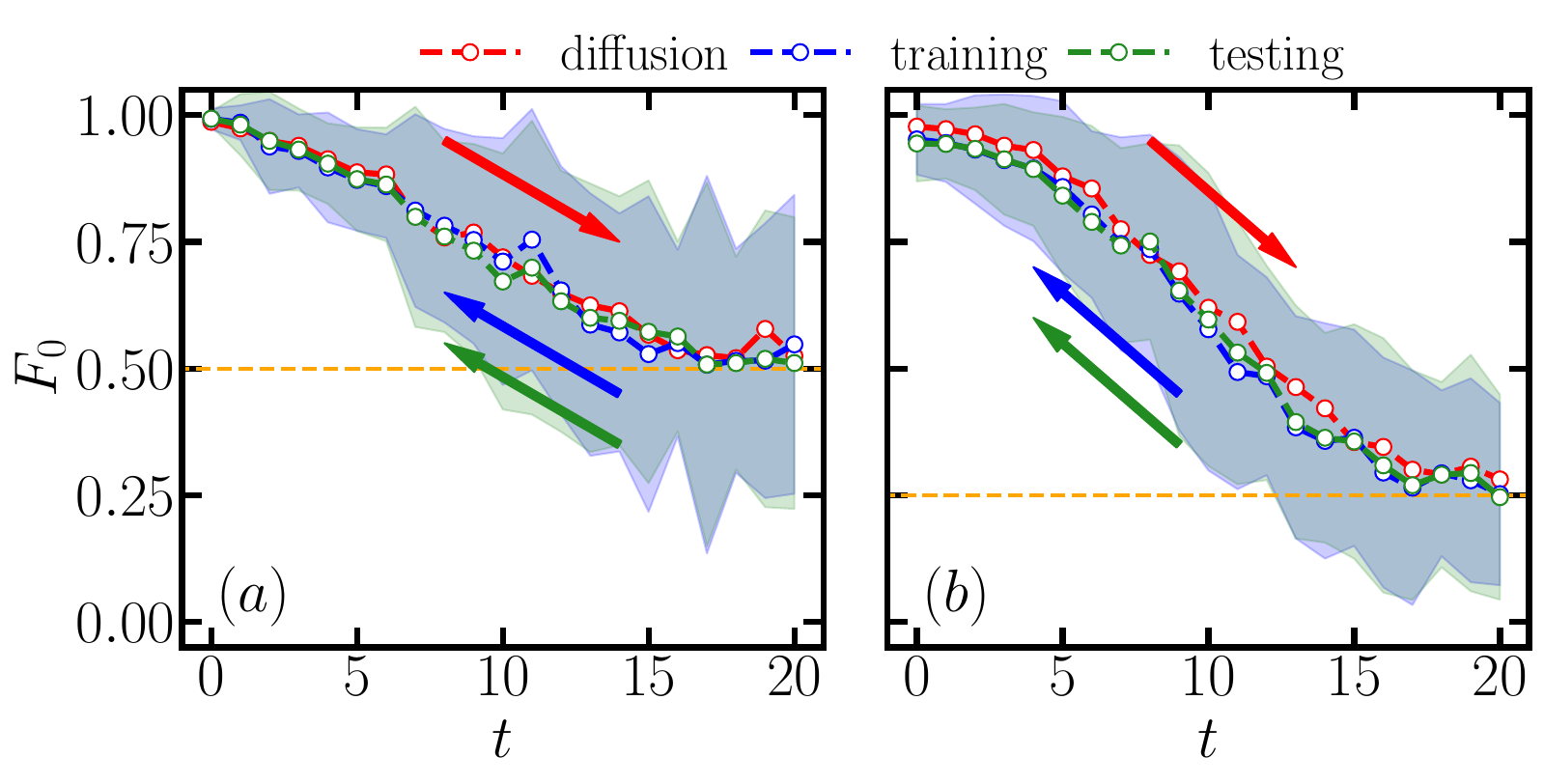}
    \caption{Dynamics of fidelity between generated states and $\ket{0}^{\otimes n}$ with (a) $n=1$ and (b) $n=2$ qubits. Average fidelity of diffusion (red), training (blue), testing (green) are plotted here with blue and green shaded area represent the standard deviation. }
    \label{fig:cluster_supp}
\end{figure}

To understand the performance more directly, besides the cost function, we can utilize the overlap fidelity $F_0$ with the state $\ket{0}$ due to the clustering of the target data. 
In Fig.~\ref{fig:cluster_supp}(a), the clustered data are same as in Fig.~\ref{fig:schematic}(b),(c) of the main text and Fig.~\ref{fig:MMD_cost_cluster} mentioned above, and the forward scrambling decreases the fidelity (red dots) from unity towards the Haar random value of $\overline{F_0}=1/2$ (yellow dashed) at $t=20$. The backward denoising (blue dots) on the other hand reverses the fidelity evolution from $t=20$ to $t=0$, which is also verified by the test data (green dots).

Now we generalize the above clustered data example to multiple qubits.
We begin with a distribution of two-qubit states clustered around $\ket{0,0}$. The data are generated in the form of $\ket{\psi^{(0)}} \sim \ket{0,0} + \epsilon \left(c_1 \ket{0,1}+c_2 \ket{1,0} + c_3\ket{1,1}\right)$ up to a normalization constant, where $\{c_i\}_{i=1}^3$ are all complex Gaussian distributed and the scale factor is chosen $\epsilon = 0.06$. We plot the fidelity with the state $\ket{0,0}$ in Fig.~\ref{fig:cluster_supp}(b) and the same evolution between ideal unity and Haar value of $1/4$ is confirmed. The MMD distance characterization is shown in Fig.~\ref{fig:MMD_cost_cluster_two_qubit} of the main text. In Fig.~\ref{fig:loss}, we show the training history of loss function in generating cluster state ensemble close to $\ket{0,0}$. The fluctuation in the training history of every denoising step is due to the random measurement in each training steps. In sample generation with optimized denoising PQC, the fluctuation is comparably small to the mean.

\section{Benchmarks: QuDT and QuGAN}
\label{app:benchmark}

\begin{figure}
    \centering
    \includegraphics[width=0.45\textwidth]{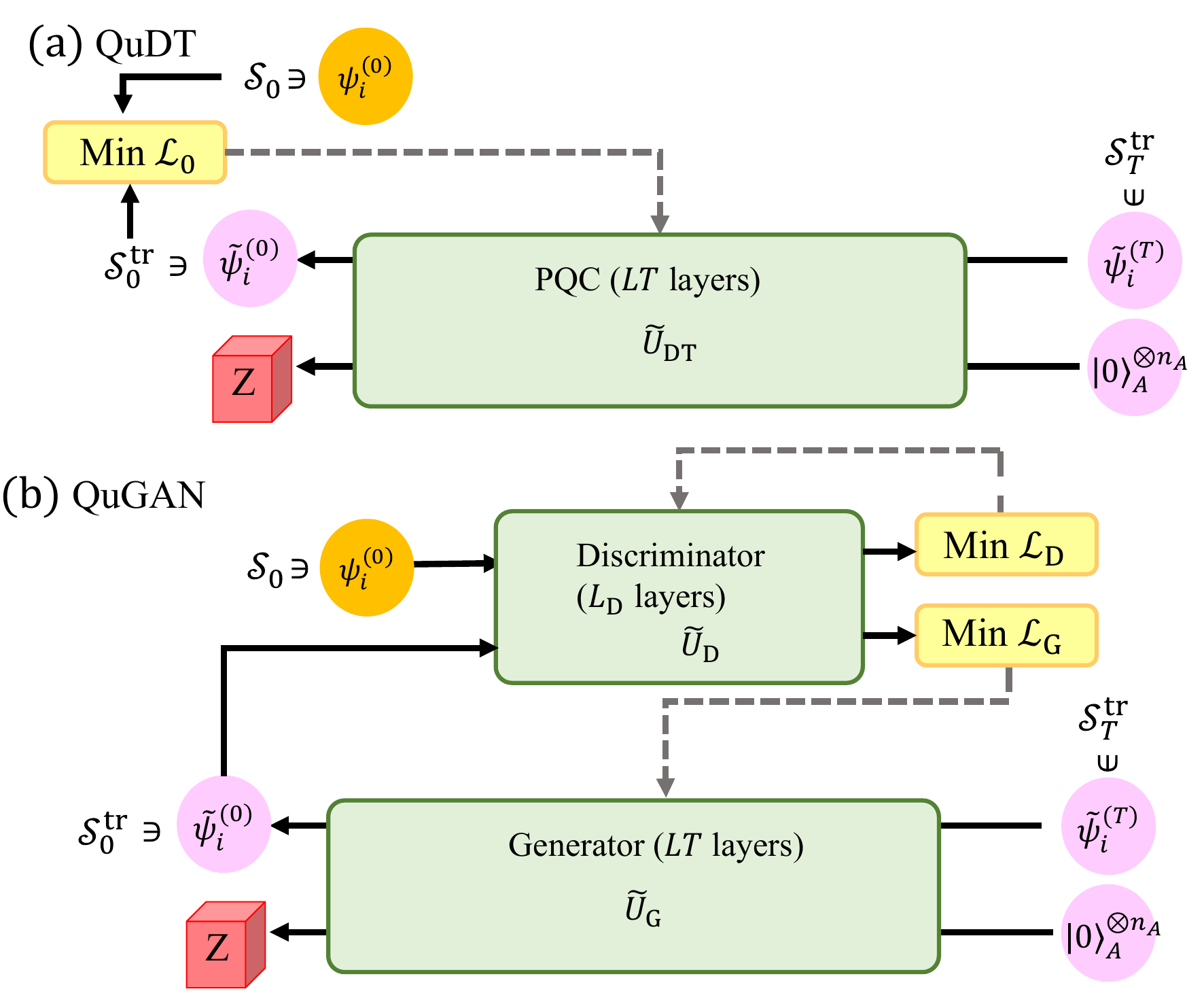}
    \caption{Schematic of QuDT (a) and QuGAN (b) in state ensemble generation. Since the forward of QuDT only has a single step from data to noise, it is not necessary to implement the diffusion for QuDT. In the backward process, there is also a single PQC $\tilde{U}_{\rm DT}$ with depth $L T$. In the training, the samples generated via applying $\tilde{U}_{\rm DT}$ on random states are directly compared to the target ensemble $\calS_0$, shown in (a). The QuGAN is similar to QuDT, except that a discriminator circuit is added to evaluate the cost function. \label{fig:qdt}
    }
\end{figure}

\begin{figure*}
    \centering
    \includegraphics[width=0.75\textwidth]{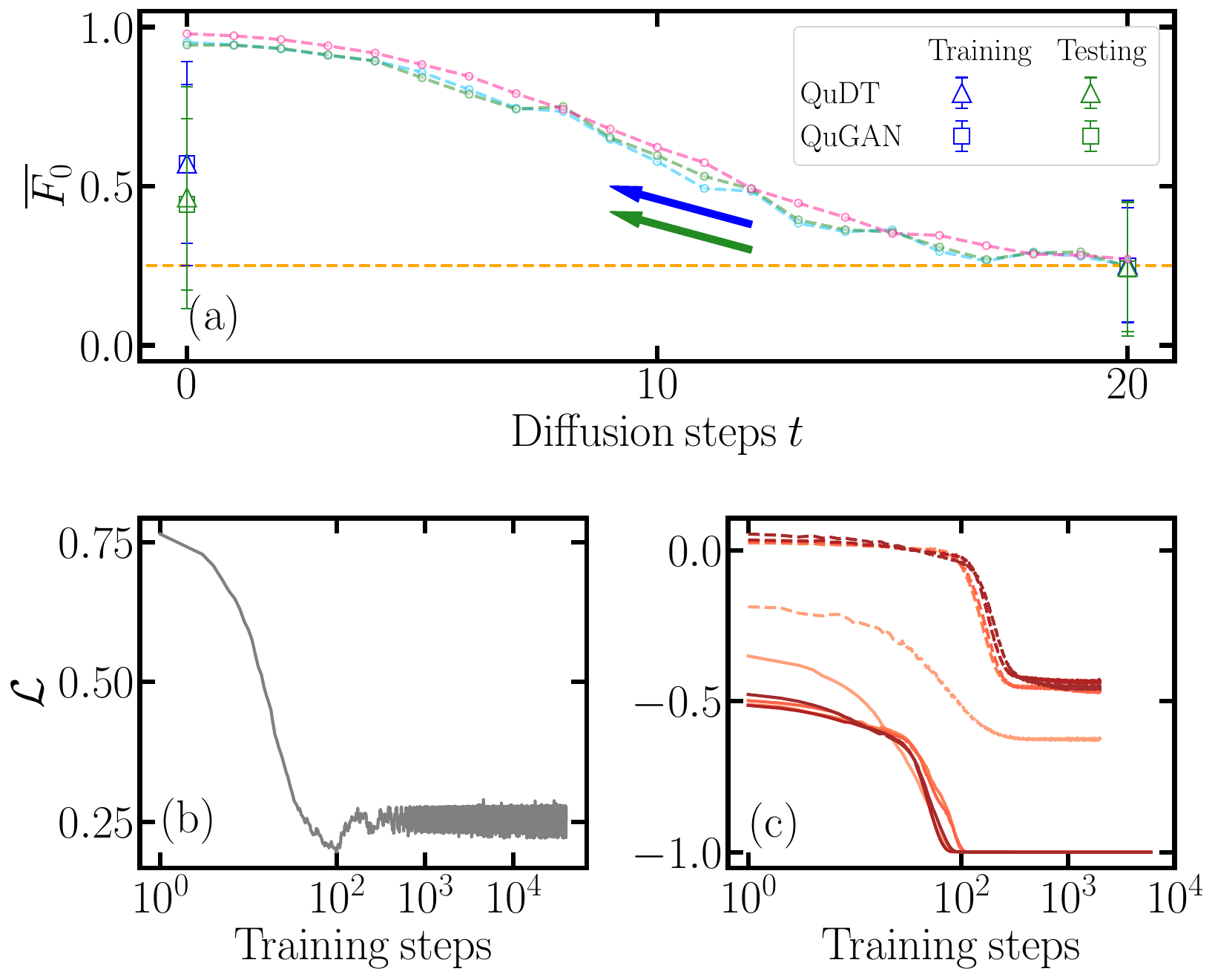}
    \caption{Benchmark between QuDDPM, QuDT and QuGAN on clustered state generation, where the true data set cluster around $\ket{0,0}$. (a) We plot mean fidelity of generated states with respect to $\ket{0,0}$ to compare their performance. Light red, blue, green curves are results from QuDDPM diffusion, training and testing for comparison. In (b) we show the QuDT training history of loss function. In (c) we show the QuGAN training history of discriminator (dashed lines) and generator (solid lines). Colors from light to dark represent the training of first to last adversarial cycle.  \label{fig:benchmark}
    }
\end{figure*}

In this section, we provide details on the benchmarks, including quantum direct transport (QuDT) and quantum generative adversarial network (QuGAN). First, we emphasize that none of the previous works actually solves the generation of an ensemble of quantum states, as we summarize in Table~\ref{tab:prior_works}. To perform the benchmark, we actually generalizes the previous works as we detail below.

\begin{table*}[]
    \centering
    \begin{tabular}{|c|c|c|}
    \hline
       Model & Task & Ref\\
    \hline
        QuGAN & Generation of a quantum state towards a given target state & ~\cite{hu2019quantum, huang2021quantum, niu2022entangling}\\ 
    \hline
        QuGAN & Learning and loading a classical distribution to quantum state & ~\cite{zoufal2019quantum, huang2021experimental, zhu2022generative}\\
    \hline
        QuCBM & Generation of a quantum state towards a target state & ~\cite{benedetti2019generative}\\
    \hline
        QuCBM & Generative learning of a classical distribution & ~\cite{gili2023quantum, liu2018differentiable}\\
    \hline
    \end{tabular}
    \caption{Summary of prior work results on quantum generative learning, including quantum generative adversarial network (QuGAN) and quantum circuit born machine (QuCBM).}
    \label{tab:prior_works}
\end{table*}

QuDT utilizes the same setup as QuDDPM, but attempts to generate the target state ensemble from a random state ensemble directly by only a single step of training on a quantum circuit with $L_{\rm DT}$ layers (see Fig.~\ref{fig:qdt}(a)), instead of discretized diffusion steps. Note that due to the single-step operation, it is not necessary to implement the forward noisy diffusion process. In the training, the loss function to be minized is
\begin{align}
    \calL_{\rm DT}(\tilde{U}_{\rm DT},\calE_0, \tilde{\calE}_T)=\calD\left(\calE_0, \tilde{\calE}_0[\tilde{U}_{\rm DT},\tilde{\calE}_T]\right),
\end{align}
which is the same as the one utilized in QuDDPM. The QuDT we proposed here can be regarded as a generalized version of quantum circuit Born machine (QuCBM)~\cite{gili2023quantum, liu2018differentiable}, which is initially targeted at generating a classical distribution from its overlap with each basis state (for example, computational basis).

Another quantum generative model to be compared with is QuGAN~\cite{lloyd2018quantum, zoufal2019quantum, hu2019quantum}. Since prior studies only focus on generating a single target state, to fit to our state ensemble generation tasks, we generalize the design of QuGAN, shown in Fig.~\ref{fig:qdt}(b). The QuGAN consists of two PQCs, a generator $\tilde{U}_{\rm G}$ and a discriminator $\tilde{U}_{\rm D}$. The generator takes the same input as the one in QuDT, while instead of minimizing a loss function directly, it tries to pass the evaluation of the discriminator that tells whether a given state comes from the real or fake state ensemble. Here we choose the discriminator having the same circuit architecture as the generator, and a single qubit Pauli-Z measurement is performed at the end to tell real/fake states apart. The training consists of several adversarial cycles, where each cycle includes the training for the discriminator and for the generator. First, we train the discriminator $\tilde{U}_{\rm D}$ while keeping the generator fixed, and the loss function to be minimized is
\begin{align}
    \calL_{\rm GAN-D} = P(\rm real|\rm fake) - P(\rm real|\rm real).
\end{align}
where $P(\rm real|\rm fake)$ is the average probability that the discriminator identifies inputs from generator as states from the true ensemble, and similar definition holds for $P(\rm real|\rm real)$. With the optimal discriminator, the probability $P(\rm real|\rm real)$ should approach unity.
Next, we train the generator $\tilde{U}_{\rm G}$ while keeping the discriminator fixed, and the loss function to be minimized is
\begin{align}
    \calL_{\rm GAN-G} = - P(\rm real|\rm fake).
\end{align}
With the optimal generator, the probability $P(\rm real|\rm fake)$ should approach unity.
Therefore, the joint optimum of generator and discriminator should lead to a near zero discriminator loss $\calL_{\rm GAN-D} \simeq 0$.

For a fair comparison among the three models, we choose the number of variational parameters in QDDPM, QuDT, and QuGAN's generator $\tilde{U}_{\rm G}$ to be the same, and the number of total training steps is also the same. In QuDDPM, we have $T=20$ diffusion steps where in each step, the layer of PQC is $L=6$. In QuDT, the layer of PQC is $L_{\rm DT} = 120$, and in QuGAN, the layer of generator is also $L_{\rm GAN-G} = 120$ and the layer of discriminator is $L_{\rm GAN-D} = 16$. We train the QuGAN with five adversarial cycles. 

\begin{figure}
    \centering
    \includegraphics[width=0.45\textwidth]{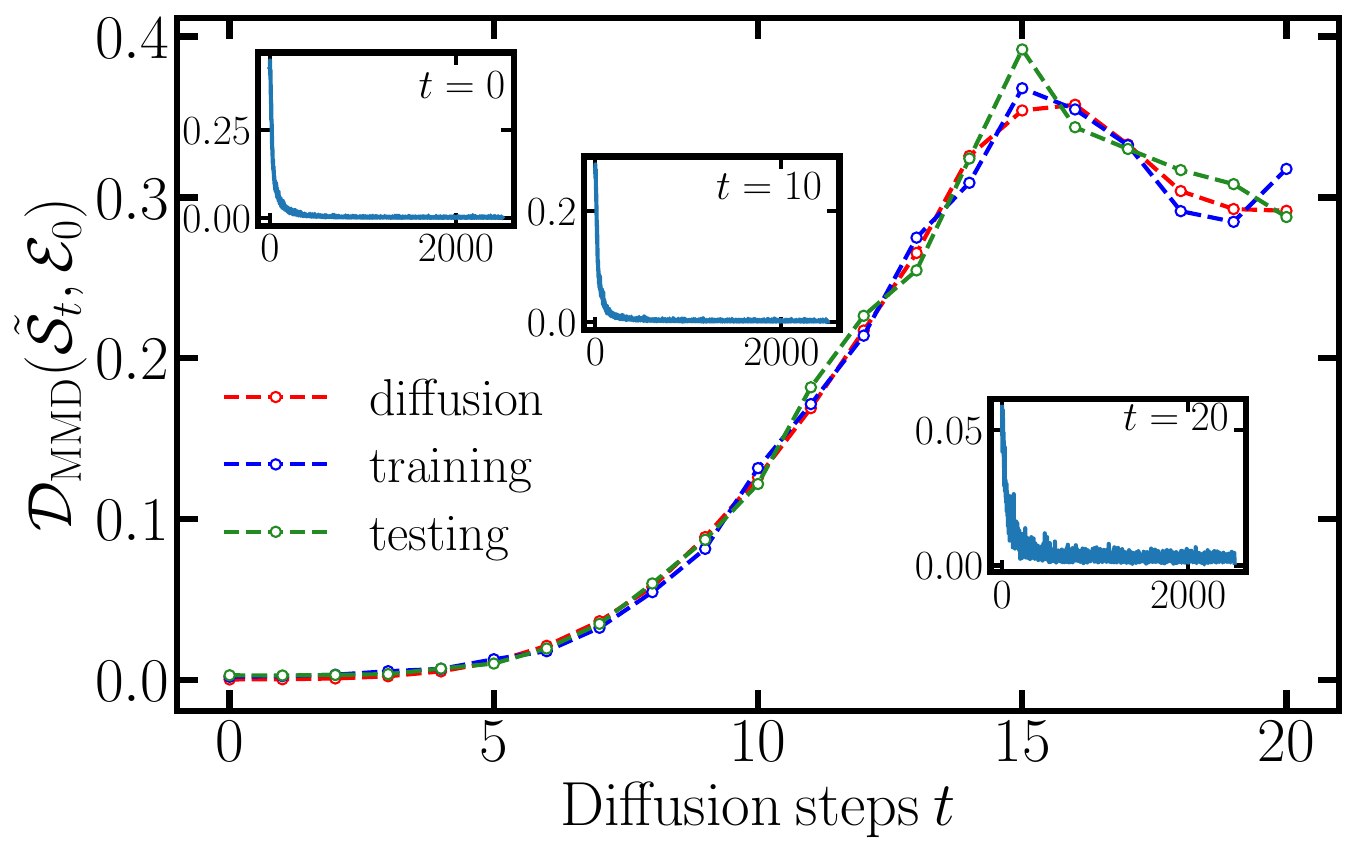}
    \caption{The decay of MMD distance $\calD$ between generated ensemble $\tilde{\calS}_t$ and target ensemble of states $\calE_0$ versus diffusion steps. The target is a state under probabilistic correlated noise (see Fig.~\ref{fig:qddpm_cluster}(a) in the main paper). Insets show the training loss history at different $t$. The converged distance at $t=0$ is nearly $2\times 10^{-3}$.}
    \label{fig:MMD_cost_cluster_error}
    \centering
    \includegraphics[width=0.45\textwidth]{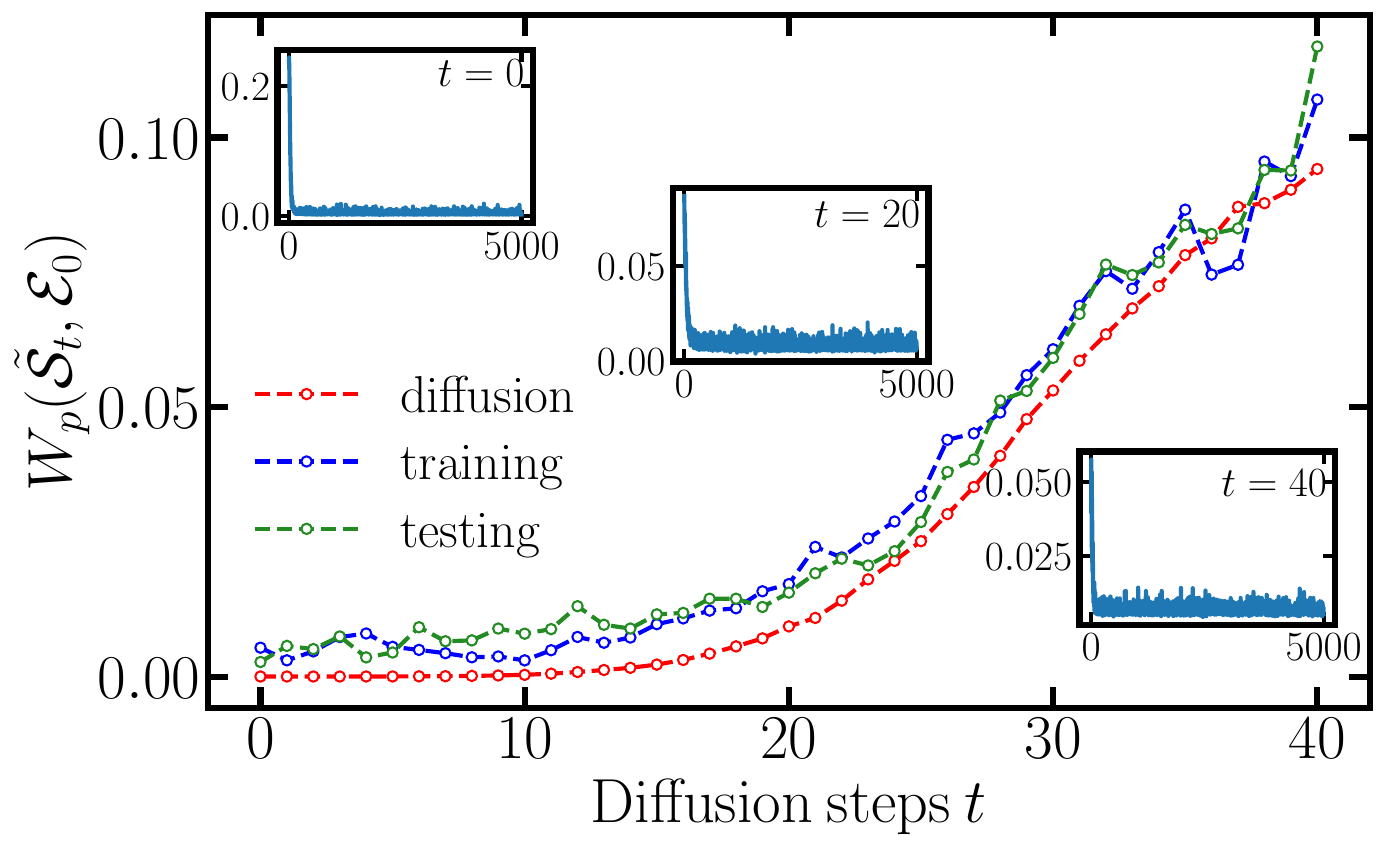}
    \caption{The decay of Wasserstein distance $W_p$ with $p=1$ between generated ensemble $\tilde{\calS}_t$ and target ensemble of states $\calE_0$ versus diffusion steps. The target is one-qubit states forming a unit circle in $X$-$Z$ plane of Bloch sphere (see Fig.~\ref{fig:qddpm_circle} of the main paper). Insets show the training loss history at different $t$. The converged distance at $t=0$ is nearly $2\times 10^{-3}$.}
    \label{fig:W_cost_circle}
\end{figure}

The MMD loss comparison is presented in Fig.~\ref{fig:MMD_cost_cluster_two_qubit} of the main text. Here we present more details on the training history and fidelity comparison.
In Fig.~\ref{fig:benchmark}, we compare the performance of QuDT and QuGAN to QuDDPM in generation of clustered state around $\ket{0, 0}$ in terms of mean fidelity which is defined as $\overline{F_0} = \mathbb{E}_{\ket{\psi} \sim \tilde{\calS}} |\braket{0,0|\psi}|^2$. As the target ensemble is clustered around $\ket{0, 0}$, a high fidelity indicates the success learning of the target ensemble. We find that compared to QuDDPM, both QuDT and QuGAN fail to generate the target ensemble of states clustered  around $\ket{0,0}$, as shown by their low fidelity: for QuDT, the fidelity converges to $\overline{F_{0,\rm tr}}=0.572\pm 0.321, \overline{F_{0,\rm te}}=0.465\pm 0.349$ for training and testing; and for QuGAN, the fidelity converges to $\overline{F_{0,\rm tr}}=0.570\pm 0.250, \overline{F_{0,\rm te}}=0.443\pm 0.249$ for training and testing. Their training history of the cost function is shown in Fig.~\ref{fig:benchmark}(b)(c), where convergence is found within the same amount of training time as QuDDPM. Thus, we believe that QuDDPM can provide advantage in the state ensemble generative learning due to its unique diffusion and denoising process.

For QuDT, the converged cost function based on MMD distance is large in Fig.~\ref{fig:benchmark}(b). This indicates QuDT is likely to be trapped in local minimum, which is typical in our various different training attempts of QuDT. The reason behind is that without the intermediate diffusion-denoising steps, the training task from full noise to a clustered distribution is challenging. For QuGAN, possible reason for failure is saddle points.
Due to the saddle point optimization nature between generator and discriminator, classical GAN training is known to have several common failure modes such as mode collapse, non-convergence and instability. For instance, if the generator finds a local strategy that can sufficiently fool the discriminator, then GAN is unable to explore beyond that local region, giving arise to mode collapse with limited diversity in the generated samples.



\section{Distance characterization in QuDDPM applications}
\label{app:distance_app}

In this section, we provide the distance characterization in the training cycles of QuDDPM's applications, similar to Fig.~\ref{fig:MMD_cost_cluster} for the simple clustered state case. In Fig.~\ref{fig:MMD_cost_cluster_error}, we show the decay of MMD distance in learning the probabilistic correlated noise and we see the distance trajectories of training and testing allign with that of the forward diffusion. In Fig.~\ref{fig:W_cost_circle}, we show the decay of Wasserstein distance $W_p$ with $p=1$ in learning the circle ensemble to characterize the difference between the generated ensemble through denoising and the target ensemble. Given the relatively small magnitude of $W_p$, we still see the agreement in denoising trajectories with respect to diffusion though with a relatively small gap.


%

\end{document}